\documentclass{article}
\usepackage[T1]{fontenc} 
\usepackage[utf8]{inputenc}
\usepackage[english]{babel}
\usepackage{fullpage}
\usepackage{url}
\usepackage{epsfig,xspace}
\usepackage[noend]{algorithmic}
\usepackage{color}
\usepackage{amsbsy,latexsym}
\usepackage{amsfonts,amsmath,amssymb}
\usepackage{url}
\usepackage{multirow, colortbl}
\usepackage[T1]{fontenc}
\usepackage{enumerate}
\usepackage[algoruled,vlined,english]{algorithm2e} 
\usepackage[utf8]{inputenc}
\usepackage{float}
\usepackage{makecell,pict2e}
\usepackage{amsthm}
\newtheorem{lemma}{Lemma}
\newtheorem{remark}{Remark}
\newtheorem{theorem}{Theorem}
\newtheorem{corollary}{Corollary}
\newtheorem{proposition}{Proposition}
\RequirePackage{fix-cm}

\newcommand{\ifConferenceVersion}{\iftrue}
\newcommand{\ifJournalVersion}{\iffalse}

\newcommand{\remove}[1]{}

 \newenvironment{proofof}[1]{\medskip\noindent {\bf Proof of #1 : }} { \qed \medskip }

\newcommand{\inst}[1]{$^{#1}$}

\newcommand{\probleme}[3]{\medskip\noindent\textbf{#1 problem}\\ \noindent \emph{Instance: }#2\\ \noindent \emph{Question: }#3\medskip}
\newcommand{\cA}{\ensuremath{\mathcal{S}\xspace}}

\newcommand\X{X_{LR}^b\xspace}
\newcommand\Y{Y_{RL}^b\xspace}
\newcommand\rclr{Reach_{LR}^c\xspace}
\newcommand\rcrl{Reach_{RL}^c\xspace}
\newcommand\rblr{Reach_{LR}^b\xspace}
\newcommand\rbrl{Reach_{RL}^b\xspace}
\newcommand\acrl{Act_{RL}^c\xspace}
\newcommand\aclr{Act_{LR}^c\xspace}
\newcommand\abrl{Act_{RL}^b\xspace}
\newcommand\ablr{Act_{LR}^b\xspace}
\newcommand\tclr{TH_{LR}^c\xspace}
\newcommand\tcrl{TH_{RL}^c\xspace}
\newcommand\sclr{S_{LR}^c\xspace}
\newcommand\scrl{S_{RL}^c\xspace}

\newcommand\Optbc{\mbox{{\tt OptimalBroadcast}}\xspace}

\newcommand\TwoAprTree{\mbox{{\tt UnknownTree}}\xspace}
\newcommand\CompThLR{\mbox{{\tt ThresholdLR}}\xspace}
\newcommand\CompOptimal{\mbox{{\tt ComputeOptimal}}\xspace}
\newcommand\CompOptimalPos{{{\tt Optimal\-At\-Index}}\xspace}
\newcommand\CompThRL{\mbox{{\tt ThresholdRL}}\xspace}
\newcommand\STree{\mbox{{\tt KnownGraph}}\xspace}

\newcommand\strat{\mbox{strategy}\xspace}
\newcommand\convergecast{convergecast\xspace}
\newcommand\broadcast{broadcast\xspace}

\newcommand\ccast{convergecast\xspace}

\newcommand\cccast{centralized convergecast\xspace}
\newcommand\Cccast{Centralized convergecast\xspace}
\newcommand\Cbcast{Centralized broadcast\xspace}
\newcommand\Dccast{Distributed convergecast\xspace}

\newcommand\subproblem{carry\xspace}
\newcommand\cS{\ensuremath{\mathcal{S}\xspace}}

\begin{document}
\sloppy

\SetAlFnt{\small\tt}
\SetAlCapFnt{\tt}
\SetProcArgSty{texttt}

\title{\bf {Convergecast and Broadcast by Power-Aware Mobile Agents\thanks{A preliminary version of this paper appeared in Proc. 26th International Symposium of Distributed Computing (DISC 2012).}}}


\author{
Julian Anaya\inst{1},
J\'{e}r\'{e}mie Chalopin\inst{2},
Jurek Czyzowicz\inst{1},\\
Arnaud Labourel\inst{2},
Andrzej Pelc\inst{1}$^,$\footnote{{Partially supported by NSERC discovery grant and by the Research Chair in Distributed Computing at the Universit\'e du Qu\'{e}bec en Outaouais.}} ,
Yann Vax\`es\inst{2}\\
\inst{1}Universit\'{e} du Qu\'{e}bec en Outaouais,
C.P. 1250, Gatineau, Qc. J8X 3X7
Canada.\\
E-mails: \url{ingjuliananaya@gmail.com},  \url{jurek@uqo.ca}, \url{pelc@uqo.ca}\\
\inst{2} LIF, CNRS \& Aix-Marseille University,
13288 Marseille, France.\\
E-mails: \url{{jeremie.chalopin,arnaud.labourel,yann.vaxes}@lif.univ-mrs.fr}
}

\maketitle

\begin{abstract}
A set of identical, mobile agents is deployed in a weighted network. Each agent has a battery -- a power source allowing it to move along network edges. An agent uses its battery proportionally to the distance traveled. We consider two tasks : {\em \convergecast}, in which at the beginning, each agent has some initial piece of information, and information of all agents has to be collected by some agent; and {\em \broadcast} in which information of one specified agent has to be made available to all other agents. In both tasks, the agents exchange the currently possessed information when they meet.

The objective of this paper is to investigate what is the minimal value of power, initially available to all agents, so that {\convergecast} or {\broadcast} can be achieved. We study this question in the centralized and the distributed settings. In the centralized setting, there is a central monitor that schedules the moves of all agents. In the distributed setting every agent has to perform an algorithm being unaware of the network. 

In the centralized setting, we give a linear-time algorithm to compute the optimal battery power and the strategy using it, both for \convergecast and for \broadcast, when agents are on the line. We also show that finding the optimal battery power for \convergecast or for \broadcast is NP-hard for the class of trees. On the other hand, we give a polynomial algorithm that finds a 2-approximation for \convergecast and a 4-approximation for \broadcast, for arbitrary graphs.

In the distributed setting, we give a 2-competitive 
algorithm for {\convergecast} in trees and a 4-competitive algorithm for \broadcast in trees. The competitive ratio of 2 is proved to be the best for the problem of \convergecast, even if we only consider line networks. Indeed, we show that there is no ($2-\epsilon$)-competitive algorithm for \convergecast or for \broadcast in the class of lines, for any $\epsilon>0$.
\end{abstract}

\section{Introduction}\label{sect:intro}

\subsection{The model and the problem}

A set of agents is deployed in a network represented by a weighted graph $G$. An edge weight is a positive real representing the length of the edge, i.e., the distance between its endpoints along the edge. The agents start simultaneously at different nodes of $G$.  Every agent has a battery: a power source allowing it to move in a continuous way along the network edges. An agent may stop at any point of a network edge (i.e. at any distance from the edge endpoints, up to the edge weight). The movements of an agent use its battery proportionally to the distance traveled. We assume that all agents move at the same speed that is equal to one, i.e., we can interchange the notions of the distance traveled and the time spent while traveling. In the beginning, the agents start with the same amount of power noted $P$, allowing all agents to travel the same distance $P$. 

We consider two tasks: {\em \convergecast}, in which at the beginning, each agent has some initial piece of information, and information of all agents has to be collected by some agent, not necessarily predetermined; and {\em \broadcast} in which information of one specified agent has to be made available to all other agents. In both tasks, agents notice when they meet (at a node or inside an edge) and they exchange the currently held information at every meeting.

 The task of \convergecast is important, e.g., when agents have partial information about the topology of the network and
the aggregate information can be used to construct a map of it, or when individual agents hold measurements
performed by sensors located at their initial positions and collected information serves to make some global decision  based on all measurements. The task of \broadcast is used, e.g., when a preselected leader has to share some information with others agents in order to organize their collaboration in future tasks.

Agents try to cooperate so that {\convergecast} (respectively \broadcast) is achieved with the smallest possible agent's initial battery power $P_{OPT}^c$ (respectively $P_{OPT}^b$), i.e., minimizing the maximum distance traveled by an agent. We investigate these two problems in two possible settings, centralized and distributed. 

In the centralized setting, the optimization problems must be solved by a central authority knowing the network and the initial positions of all the agents. We call \emph{\strat} a finite sequence of movements executed by the agents. During each movement, starting at a specific time, an agent walks between two points belonging to the same network edge. A {\strat} is a {\convergecast} {\strat} if the sequence of movements results in one agent getting the initial information of every agent. 
A {\strat} is a {\broadcast} {\strat} if the sequence of movements results in all agents getting the initial information of the source agent. We consider two different versions of the problem : the decision problem, i.e., deciding if there exists a   {\convergecast} {\strat} or a broadcast \strat using power $P$ (where $P$ is the input of the problem) and the optimization problem, i.e., computing the smallest amount of power that is sufficient to achieve \convergecast or \broadcast. 

In the distributed setting, the task of \convergecast or broadcast must be approached individually by each agent.
Each agent is unaware of the network, of its position in the network and of the positions (or even the presence) of any other agents. The agents are anonymous and they execute the same deterministic algorithm. Each agent has a very simple sensing device allowing it to detect the presence of other agents at its current location in the network. Each agent is also aware of the degree of the node at which it is located, as well as the port through which it enters a node, called an \emph{entry port}. We assume that the ports of a node of degree $d$ are represented by integers $1,2, \dots, d$. Agents can meet at a node or inside an edge. When two or more agents meet at a node, each of them is aware of the direction from which the other agent is coming, i.e., the last entry port of each agent.

 Since the measure of efficiency in this paper is the battery power (or the maximum distance traveled by an agent, which is proportional to the battery power used) we do not try to optimize the other resources (e.g. global execution time, local computation time, memory size of the agents, communication bandwidth, etc.). In particular, we conservatively suppose that, whenever two agents meet, they automatically exchange the entire information they hold (rather than the new information only). This information exchange procedure is never explicitly mentioned in our algorithms, supposing, by default, that it always takes place when a meeting occurs. The efficiency of a distributed solution is expressed by the {\em competitive ratio}, which is the worst-case ratio of the amount of power necessary to solve the {\convergecast} or the broadcast problem by the distributed algorithm with respect to the amount of power computed by the optimal centralized algorithm, which is executed for the same agents' initial positions. 

It is easy to see, that in the optimal centralized solution for the case of the line and the tree, the original network may be truncated by removing some portions and leaving only the connected part of it containing all the agents (this way all leaves of the remaining tree contain initial positions of agents). We make this assumption also in the distributed setting, since no finite competitive ratio is achievable if this condition is dropped. Indeed, two nearby anonymous agents inside a long line need to travel, in the worst case, a long distance to one of its endpoints in order to meet.

\subsection{Related work}\label{s:related}

Rapidly developing network and computer industry fueled the research interest in mobile agents computing. Mobile agents are often interpreted as 
software agents, i.e., programs migrating from host to host in a network, performing some specific tasks. However, the recent developments in computer technology bring up problems related to physical mobile devices. These include robots or motor vehicles and various wireless gadgets.
Examples of agents also include living beings: humans (e.g. soldiers in the battlefield or disaster relief personnel) or animals (e.g. birds, swarms of insects). 

In many applications the involved mobile agents are small and have to be produced at low cost in massive numbers. Consequently, in many papers, the computational power of mobile agents is assumed to be very limited and feasibility of some important distributed tasks for such collections of agents is investigated. For example \cite{AA06} introduced {\em population protocols}, modeling wireless sensor networks by extremely limited finite-state computational devices. The agents of population protocols move according to some mobility pattern totally out of their control and they interact randomly in pairs. This is called {\em passive mobility}, intended to model, e.g., some unstable environment, like a flow of water, chemical solution, human blood, wind or unpredictable mobility of agents' carriers (e.g. vehicles or flocks of birds). On the other hand,
\cite{SY} introduced anonymous, oblivious, asynchronous, mobile agents which cannot directly communicate, but they can occasionally observe the environment. Gathering and convergence  \cite{AOSY,CFPS,CP,C15}, as well as pattern formation \cite{DFSY,FPSW,SY,YS} were studied for such agents.

Apart from the feasibility questions for limited agents, the optimization problems related to the efficient usage of agents' resources have been also investigated. Energy management of (not necessarily mobile) computational devices has been a major concern in recent research papers (cf. \cite{Albers}). Fundamental techniques proposed to reduce power consumption of computer systems include power-down strategies (see \cite{Albers,AIS,ISG}) and speed scaling (introduced in \cite{YDS}). Several papers proposed centralized \cite{Bunde,SL,YDS} or distributed  \cite{Albers,Ambuhl,AIS,ISG} algorithms. However, most of this research on power efficiency concerned optimization of overall power used. Similar to our setting, assignment of charges to the system components in order to minimize the maximal charge has a flavor of another important optimization problem which is load balancing (cf. \cite{Azar}).

In wireless sensor and ad hoc networks the power awareness has been often related to the data communication via efficient routing protocols (e.g. \cite{Ambuhl,SL}. However in many applications of mobile agents (e.g. those involving actively mobile, physical agents) the agent's energy is mostly used for it's mobility purpose rather than communication, since active moving often requires running some mechanical components, while communication mostly involves (less energy-prone) electronic devices.
Consequently, in most tasks involving moving agents, like exploration, searching or pattern formation, the distance traveled is the main optimization criterion (cf. \cite{AH,AG,ABRS,BCR,BeRS,BlRS,DP,DKS,FGKP,MMS}). Single agent exploration of an unknown environment has been studied for graphs, e.g. \cite{AH,DP}, or geometric terrains, \cite{BCR,BlRS}. 

While a single agent cannot explore a graph of unknown size unless pebble (landmark) usage is permitted (see \cite{BFRSV}), a pair of robots are able to explore and map a directed graph of maximal degree $d$ in $O(d^2n^5)$ time with high probability (cf. \cite{BS}). In the case of a team of collaborating mobile agents, the challenge is to balance the workload among the agents so that the time to achieve the required goal is minimized. However this task is often hard (cf. \cite{FHK}), even in the case of two agents in a tree, \cite{AB}. On the other hand, the authors of \cite{FGKP} study the problem of agents exploring a tree, showing $O(k/ \log k)$ competitive ratio of their distributed algorithm provided that writing (and reading) at tree nodes is permitted. 

Assumptions similar to our paper have been made in \cite{ABRS,BlRS,DKS} where the mobile agents are constrained to travel a fixed distance to explore an unknown graph \cite{ABRS,BlRS}, or tree \cite{DKS}. In \cite{ABRS,BlRS} a mobile agent has to return to its home base to refuel (or recharge its battery) so that the same maximal distance may repeatedly be traversed. \cite{DKS} gives an 8-competitive distributed algorithm for a set of agents  with the same amount of power exploring the tree starting at the same node.

The {\convergecast} problem is sometimes viewed as a special case of the data aggregation question (e.g. \cite{KEW,RV}) and it has been studied mainly for wireless and sensor networks, where the battery power usage is an important issue (cf. \cite{KK,AGS}). Recently \cite{CJABL} considered the online and offline settings of the scheduling problem when data has to be delivered to mobile clients while they travel within the communication range of wireless stations. \cite{KK} presents a randomized distributed {\convergecast} algorithm for geometric ad-hoc networks and study the trade-off between the energy used and the latency of {\convergecast}. The \broadcast problem for stationary processors has been extensively studied both for the message passing model, see e.g. \cite{AGP}, and for the wireless model, see e.g. \cite{BGI}. To the best of our knowledge, the problem of the present paper, when the mobile agents perform {\convergecast} or broadcast by exchanging the held information when meeting, while optimizing the maximal power used by a mobile agent, has never been investigated before.

\subsection{Our results}\label{s:results}

In the centralized setting, we give a linear-time algorithm to compute the optimal battery power and the strategy using it, both for \convergecast and for \broadcast, when agents are on the line. We also show that finding the optimal battery power for \convergecast or for \broadcast is NP-hard for the class of trees. In fact, the respective decision problem is strongly NP-complete. On the other hand, we give a polynomial algorithm that finds a 2-approximation for \convergecast and a 4-approximation for \broadcast, for arbitrary graphs.

In the distributed setting, we give a 2-competitive 
algorithm for {\convergecast} in trees and a 4-competitive algorithm for \broadcast in trees. The competitive ratio of 2 is proved to be the best for the problem of \convergecast, even if we only consider line networks. Indeed, we show that there is no ($2-\epsilon$)-competitive algorithm for \convergecast or for \broadcast in the class of lines, for any $\epsilon>0$.

The following table gives the summary of our results.

\begin{table}[h]
\centering
\begin{tabular}{|c|m{6cm}|m{6cm}|}\hline
\diaghead{Problems Setting}%
{\large Setting}{\large Problems}&
\multicolumn{1}{c|}{Convergecast}&\multicolumn{1}{c|}{Broadcast}\\
\hline
\multirow{3}{*}[-0.22cm]{Centralized}&\multicolumn{2}{m{12cm}|}{$\bullet$ linear-time algorithm to compute optimal battery power and strategy on lines}\\
&\multicolumn{2}{l|}{$\bullet$ proof that the above problem is NP-hard on trees}\\\cline{2-3}
&$\bullet$ polynomial 2-approximation on arbitrary graphs&$\bullet$ polynomial 4-approximation on arbitrary graphs\\\hline
\multirow{3}{*}[0.22cm]{Distributed}&$\bullet$ 2-competitive algorithm for trees&$\bullet$ 4-competitive algorithm for trees\\\cline{2-3}
&\multicolumn{2}{m{8cm}|}{$\bullet$ proof that there is no $(2-\epsilon)$-competitive algorithm on lines, for any $\epsilon>0$}\\\hline
\end{tabular}
\caption{Summary of our results}
\end{table}

\noindent\textbf{Roadmap}

In Section 2, we show that we can restrict the search for the optimal strategy for convergecast or broadcast on the line to some smaller subclass of strategies called regular strategies. In Section 3, we present our centralized algorithms for convergecast and broadcast on lines. Section 4 is devoted to centralized convergecast and broadcast on trees and graphs. In Section 5, we investigate convergecast and broadcast in the distributed setting. Section 6 contains conclusions and open problems.

\section{Regular strategies for convergecast and broadcast on lines}

In this section, we show that if we are given a {\convergecast} (respectively \broadcast) strategy for some initial positions of agents in the line, then we can always modify it in order to get another {\convergecast} (respectively \broadcast) strategy, using the same amount of maximal power for every agent, satisfying some simple properties. Such strategies will be called \emph{regular}. These observations permit to restrict the search for the optimal strategy to some smaller and easier to handle subclass of
strategies.

We order agents according to their positions on the line. Hence we can assume w.l.o.g., that agent $a_i$, for $1 \leq i \leq n$ is initially positioned at point $Pos[i]$ of the line of length $\ell$ and that $0\leq Pos[1]< Pos[2] < \ldots < Pos[n]\leq\ell $. 
The set $Pos[1:n]$ will be called a \emph{configuration} for the line of length $\ell$.

\subsection{Regular {\subproblem} strategies}

Given a configuration $Pos[1:n]$, a starting point $s$, a
target point $t$ ($s < t$), and an amount of power $P$, we want to know if there
exists a strategy $\cS$ for the agents enabling them to move the
information from $s$ to $t$ so that the amount of power spent by
each agent is at most $P$. Strategies that move information from point $s$ to point $t$ will be called \emph{carry} strategies for $(Pos[1:n],s,t,P)$. We restrict attention to configurations $Pos[1:n]$ such that $|s -Pos[1]|<P$ and $|t-Pos[n]|<P$ because otherwise either $Pos[1]$ (respectively $Pos[n]$) is useless or it is impossible to carry information from $s$ to $t$. A \emph{regular} carry strategy for $(Pos[1:n],s,t,P)$ is the set of moves for agents $a_1,a_2,\dots, a_n$ defined as follows: agent $a_i$ first goes back to a point $b_i \leq Pos[i]$, getting there the information
from the previous agent (except $a_1$ that has to go to $s$), then it goes forward to a point $f_i \geq b_i$. Moreover, we require
that each agent travels the maximal possible
distance, i.e., it spends all its power.

\begin{lemma}\label{lem-subproblem}
%
%
%
If there exists a \subproblem strategy for  $(Pos[1:n],s,t,P)$, then there exist the following two regular \subproblem strategies. 

The pull strategy that can be computed iteratively (in linear time) starting with the
last agent:
\begin{enumerate}
\item $b_1 \leq s$, $f_n = t$,
\item $b_i = f_{i-1}, \forall 2 \leq i \leq n$,
\item $f_i = P + 2b_i - Pos[i] \geq b_i,  \forall 1 \leq i \leq n$.
\end{enumerate}

The push strategy that can be computed iteratively (in linear time)  starting with the
first agent:
\begin{enumerate}
\item $b_1 = \min\{Pos[1],s\}$, $f_n \geq t$,
\item $b_i = \min(f_{i-1},Pos[i]), \forall 2 \leq i \leq n$,
\item $f_i = P + 2b_i - Pos[i] \geq b_i,  \forall 1 \leq i \leq n$.
\end{enumerate}
\end{lemma}

\begin{proof}
%
We first show that there exists a pull strategy. 
Consider $(Pos[1:n],s,t,P)$ with the minimum number of
agents such that there exists a \subproblem strategy, but no
pull strategy. We consider the smallest value $s$ such that $(Pos[1:n],s,t,P)$
admits a \subproblem strategy but no pull strategy.

If $Pos[1] < s$, then either $Pos[1]+P < s$, or $Pos[1]+P \geq s$. In the
first case, $a_1$ cannot move the information between $s$ and $t$, and
then $(Pos[2:n],s,t,P)$ admits a carry strategy but not a pull strategy and has fewer agents. In the
second case, $\cS$ is also a \subproblem strategy for $(Pos[1:n],Pos[1],t,P)$ and there is no pull strategy for $(Pos[1:n],Pos[1],t,P)$,
contradicting our choice of $s$.

Hence, we may suppose that $Pos[1]\geq s$. Since there exists a carry strategy $\cS$, let $a_i$ be the first agent that
reaches $s$.  The rightmost point where $a_i$ can move the information
from $s$ is $s'= 2s+P-Pos[i]$. Since $\cS$ is a \subproblem strategy,
when considering all the agents except $i$, $\cS$ is a \subproblem
strategy for $(Pos[S\setminus\{i\}],s',t,P)$. By minimality of the
number of agents, the pull strategy solves the subproblem on
$(Pos[S\setminus\{i\}],s',t,P)$.  Consequently, we can assume that
$\cS$ is a pull strategy on $(Pos[S\setminus\{i\}],s',t,P)$.  If $i =
1$, by minimality of $s$, we have $s' = b_2$ and thus $\cS$ is a pull
strategy which is a contradiction. Hence, suppose that $i>1$. Note that if $Pos[i]=Pos[1]$, we can exchange the roles of
$a_i$ and $a_1$ and we are in the previous case.
Hence, suppose that $Pos[i] > Pos[1]$ and let $[b_1,f_1]$ be the interval
that $a_1$ traverses with the information when $\cS$ is applied; by
minimality of $s$, $b_1 = s'$ and consequently we have $P = Pos[i] +
b_1 - 2s = Pos[1]+f_1-2b_1$, and thus $s = (2Pos[i]+Pos[1]+f_1-3P)/4$.
Consider now the strategy where we exchange the roles of $a_1$ and
$a_i$: $a_1$ gets the information from $s$, gives it to $a_i$, and
$a_i$ goes to $f_1$. More formally, let $f_i' = f_1$, $b_i' = (Pos[i]
+ f_i' - P)/2$, $f_1' = b_i'$ and $b_1' = (Pos[1] + f_1' - P)/2$. From
our definition of $f_1'$ and $s_1'$ and the first part of the proof,
there exists a \subproblem strategy for
$(Pos[1:n],b_1',t,P)$. However, $b_1' = (2Pos[1]+Pos[i]+f_1-3P)/4 = s
+ (Pos[1] - Pos[i])/4 < s$, contradicting the minimality of $s$.

Consequently, if there exists a carry strategy $\cS$ for $(Pos[1:n],s,t,P)$, then there exists
a pull strategy on
$(Pos[1:n],s,t,P)$. 

\medskip

Now suppose that $(Pos[1:n],s,t,P)$ admits a \subproblem strategy. 
From the first part of the proof, we know that it admits a pull strategy.
The push strategy for $(Pos[1:n],s,t,P)$ can
be obtained inductively from the pull strategy. 
Let $[b_i,f_i]$ for $i=1,...,n$ be the set 
of intervals that induces the pull strategy for 
$(Pos[1:n],s,t,P).$ Notice that $[b_i,f_i]$ for $i=1,...,n-1$ 
induces the pull strategy for $(Pos[1:n-1],s,b_n,P).$ 
By induction, there exists a set of intervals $[b'_i,f'_i]$ that 
induces a push strategy for $(Pos[1:n-1],s,b_n,P)$ 
with $f'_{n-1}\ge b_n.$ We define $b'_n = \min\{Pos[n],f'_{n-1}\}$ 
and $f'_n = P + 2b'_n - Pos[n].$ Since $b'_n\ge f'_{n-1} \ge b_n,$ 
we deduce that $f'_n \ge f_n \ge t$ and therefore the set of intervals
$[b'_i,f'_i]$ induces a push strategy for 
$(Pos[1:n],s,t,P)$.  
\end{proof}

\begin{remark}\label{rem-pull-push}
Note that the pull strategy is uniquely defined by a configuration
$Pos[1:n]$, a target point $t$, and an amount of power $P$ and enables
to compute the smallest $s$ such that $(Pos[1:n],s,t,P)$ admits a
\subproblem strategy.

Similarly, the push strategy is uniquely defined by a configuration
$Pos[1:n]$, a starting point $s$, and an amount of power $P$ and enables
to compute the largest $t$ such that $(Pos[1:n],s,t,P)$ admits a
\subproblem strategy.
\end{remark}

Note that carry strategies are defined for the target $t$ larger than the starting point $s$. A carry strategy will be called \emph{reverse} if the target $t$ is smaller than $s$ and all moves to the right are replaced by moves to the left and vice-versa.

\subsection{Regular convergecast strategies}

We now define the notion of a regular convergecast strategy for $Pos[1:n]$ on the segment $[0,\ell]$, using power at most $P$. 
Without loss of generality, we suppose that $Pos[1]=0$ and $Pos[n]=\ell$. Intuitively, a regular convergecast strategy divides the set of all agents into the set of left agents and the set of right agents such that left agents execute a push strategy from $Pos[1]$ and right agents execute a reverse push strategy from $Pos[n]$.

More formally, a \emph{regular} convergecast strategy is given by a partition of the agents into two sets $LR = \{a_i
\mid i \leq p\}$ and $RL = \{a_i \mid i > p\}$ for some $p$, and by two points $b_i, f_i$ of segment $[0, \ell]$
for each
agent $a_i$,
such that
\begin{enumerate}[(1)]
\item if $a_i \in LR$, $b_i = \min \{f_{i-1},Pos[i]\}$ and
  $f_i= 2b_i+P-Pos[i]$,
\item if $a_i \in RL$, $b_i = \max \{f_{i+1},Pos[i]\}$ and $f_i =
  2b_i-P-Pos[i]$, 
\item $F_{LR} = \max \{f_i \mid a_i \in LR\} \geq F_{RL} = \min\{f_i
  \mid a_i \in RL\}$.
\end{enumerate}

Suppose that we are given a partition of the agents into two
disjoint sets $LR$ and $RL$ and values $b_i, f_i$ for each
agent $a_i$ satisfying conditions (1)-(3). Then the following moves define a regular convergecast strategy: first, every agent $a_i \in LR \cup RL$ moves to $b_i$;
subsequently, every agent in $LR$ moves to $f_i$ once it learns the initial
information of $a_1$; then, every agent in $RL$ moves to $f_i$ once
it learns the initial information of $a_n$.  Let $a_k$ be an agent from $LR$
such that $f_k$ is maximum. Once $a_k$ has moved to $f_k$, it knows
the initial information of all the agents $a_i$ such that $b_i \leq
f_k$. If $f_k \geq \ell$, {\convergecast} is achieved. Otherwise,
since $f_k = \max \{f_i \mid a_i \in LR\} \geq \min\{f_i \mid a_i \in
RL\}$, we know that there exists an agent $a_j \in RL$ such that $f_j
\leq f_k < b_j$.  When $a_j$ reaches $f_k$ it knows the initial
information of all the agents such that $b_i \geq f_k$ and thus, $a_j$
and $a_k$ know the initial information of all agents, which accomplishes convergecast.

The following lemma shows that we can restrict attention to regular convergecast strategies.

\begin{lemma}\label{lem:regconv}
If there exists a {\convergecast} strategy for a configuration
$Pos[1:n]$ using power at most $P$ then there exists a regular convergecast strategy for the configuration
$Pos[1:n]$ using power at most $P$.
\end{lemma}

\begin{proof}
Consider a {\convergecast} strategy $\cA$ for a configuration
$Pos[1:n]$ using power at most $P$. Suppose that convergecast occurred at time $t$ at some point $q$.
If an agent $a_i$ does not get the initial information of $a_1$, then at time $t$ it must have been in the segment 
$[q,Pos[n]]$. Hence, by time $t$, it must have learned the initial information of $a_n$. It follows that every
agent $a_i$, for $1 < i < n$, must learn either the initial
information of agent $a_1$ or of $a_n$. Therefore, we can partition the
set of agents performing a {\convergecast} strategy into two subsets
$LR$ and $RL$, such that each agent $a_i \in LR$ learns the initial
information of agent $a_1$ before learning the initial information of
agent $a_n$ (or not learning at all the information of $a_n$). All
other agents belong to $RL$. We denote by $[b_i,f_i]$ the interval of
all points visited by $a_i \in LR$ and by $[f_j,b_j]$ the interval of points
visited by $a_j \in RL$.

Let $F_{LR} = \max \{f_i \mid a_i \in LR\}$ and $F_{RL} = \min \{f_j
\mid a_j \in RL\}$. Since $\cA$ is a convergecast strategy, we have $F_{LR} >
F_{RL}$. Observe that the agents in $LR$ move the initial information of $a_1$ from
$Pos[1]$ to $F_{LR}$ and that the agents in $RL$ move the initial information of $a_n$
from $Pos[n]$ to $F_{RL}$. From Lemma~\ref{lem-subproblem}, we can
assume that the agents in $LR$ (resp. $RL$) execute a push strategy (resp. a reverse push strategy)
and thus conditions (1)-(3) hold.

Suppose now that there exists an agent $a_i \in RL$ such that $a_{i+1}
\in LR$.  Let $f_{RL}(i) = \min \{f_j \mid a_j \in RL, j>i\}$ and
$f_{LR}(i+1)=\max \{f_j \mid a_j \in LR, j<i\}$; note that $b_i = \max
\{f_{RL}(i), Pos[i]\}$ and $b_{i+1} = \min \{f_{LR}(i+1),
Pos[i+1]\}$. Consider the strategy where we exchange the roles of
$a_i$ and $a_{i+1}$, i.e., we put $a_i \in LR$ and $a_{i+1} \in
RL$. Let $b_i' = \min\{f_{LR}(i+1),Pos[i]\}$, $b'_{i+1} =
\max\{f_{RL}(i),Pos[i+1]\}$, $f_i' = 2b_i'+P-Pos[i]$ and $f_{i+1}' =
2b_{i+1}'-P-Pos[i+1]$.

If $f_{RL}(i) \leq Pos[i+1]$, then $f'_{i+1} = Pos[{i+1}]-P \leq
b_{i+1} \leq f_{LR}(i+1)$. If $f_{LR}(i+1) \geq Pos[i]$, then $f'_{i}
= Pos[i]+P \geq b_{i} \geq f_{RL}(i)$. In both cases, we still have a
{\convergecast} strategy.

If $f_{RL}(i) \geq Pos[i+1]$ and $f_{LR}(i+1) \leq Pos[i]$, then $f_i'
= 2 f_{LR}(i+1) + P -Pos[i] > 2 f_{LR}(i+1) + P -Pos[i+1] = f_{i+1}$,
and $f_{i+1}' = 2 f_{RL}(i) - P -Pos[i+1] < 2 f_{RL}(i) - P -Pos[i] =
f_{i}$. Consequently, we still have a {\convergecast} strategy.

Applying this exchange a finite number of times, we get a regular convergecast
strategy. 
\end{proof}

\subsection{Regular {\broadcast} strategies}

We now define the notion of a regular broadcast strategy for $Pos[1:n]$ where the source agent is $a_k$, on the segment $[0,\ell]$, using power at most $P$. Without loss of generality, we suppose that $Pos[1]=0$ and $Pos[n]=\ell$. Intuitively, a regular broadcast strategy divides the set of all agents into the set of left agents and the set of right agents such that left agents execute a reverse pull strategy from $Pos[k]$ and right agents execute a pull strategy from $Pos[k]$.

More formally, a \emph{regular} broadcast strategy is given by points $b_i, f_i$ of segment $[0, \ell]$ defined for each agent $a_i$ such that
\begin{enumerate}
\item $b_1 = f_1 = Pos[1]+P$, $b_n = f_n= Pos[n] -P$,
\item if $1 < i < k$, $f_i= b_{i-1}$ and $b_i = (f_i+Pos[i]+P)/2$,
\item if $k < i < n$, $f_i= b_{i+1}$ and $b_i = (f_i+Pos[i]-P)/2$,
\item $\{b_k,f_k\} = \{b_{k-1},b_{k+1}\}$ and $|2b_k - Pos[k] - f_k|
  \leq P$ 
\end{enumerate}
Suppose that we are given points $b_i, f_i$ for each agent $a_i$, satisfying conditions (1)-(4).
Then the following moves define a regular broadcast strategy:
 initially every agent
$a_i$ moves to $b_i$. Once $a_i$ learns the source information, $a_i$ moves to $f_i$. Since (1)-(4) hold, this is a broadcast strategy and the maximum amount of power spent is at most $P$. 

Before proving that it is enough to only consider regular broadcast strategies, we need to prove the following technical lemma. 

\begin{lemma}\label{lem-strat1-broadcast}
There exists a broadcast strategy $\cS$ for a configuration
$(Pos[1:n],k,P)$ if and only if for every $i$, there exist positions
$l_i, x_i, r_i$ such that
\begin{enumerate}[(1)]
\item  for each $i$,  $l_i \leq x_i \leq r_i$
\item $x_k = Pos[k]$;
\item for each $i$, $|x_i - Pos[i]| + \min(x_i+r_i-2l_i, 2r_i-x_i-l_i)
  \leq P$.
\item for each $i$, if $x_i < Pos[k]$ (resp. $x_i > Pos[k]$), there
  exists $j$ such that $x_i \in [l_j,r_j]$ and $x_j > x_i$ (resp. $x_j
  < x_i$).
\end{enumerate}
\end{lemma}

\begin{proof}
Consider a broadcast strategy $\cS$ where the maximum amount of power
spent is $P$.  For every agent $a_i$, let $x_i$ be the position where
$a_i$ learns the information that has to be broadcast, and let $l_i$
(resp. $r_i$) be the leftmost (resp. rightmost) position reached by
$a_i$ once it got the information.  By definition of $l_i, x_i, r_i$,
(1) and (2) hold. Since the maximum amount of power spent by an agent
is at most $P$, and since the agent has to go from $Pos[i]$ to $x_i$
and then to $r_i$ and $l_i$, (3) holds.  Since every agent learns the
information, for every agent $a_i$, either $x_i = Pos[k]$, or $a_i$ meets
an agent $a_j$ in $x_i$ such that $a_j$ already has the information.
Assume that $x_i < Pos[k]$ (the other case is symmetric). If $x_i <
x_j$, then (4) holds for $i$.  Suppose now that $x_j \leq x_i \leq
Pos[k]$ and let $A$ be the non-empty set of agents $a_j$ such that
$x_j \leq x_i$ and $a_j$ learns the information before $a_i$. Let $a_j
\in A$ be the agent that is first to learn the information. Since $x_j \leq
x_i < Pos[k]$, $a_j$ learns the information from an agent $a_{j'}$
that does not belong to $A$. Consequently, $x_{j'} > x_i \geq x_j$ and
thus $x_i \in [x_j,x_{j'}] \subseteq [l_{j'},r_{j'}]$. Thus (4) holds
for $i$. 

Conversely, if we are given values $x_i, l_i, r_i$ satisfying (1)-(4),
we can exhibit a strategy for broadcast: initially every agent
$a_i$ moves to $x_i$. Once $a_i$ learns the information, if
$x_i+r_i-2l_i \leq 2r_i-x_i-l_i$, then $a_i$ moves to $l_i$ and to
$r_i$ and if $x_i+r_i-2l_i > 2r_i-x_i-l_i$, then $a_i$ moves to $r_i$ and
to $l_i$. Since (4) holds, this is a broadcast strategy and since (3) holds, the maximum amount of power spent is at most $P$. 
\end{proof}

\medskip

The following lemma shows that we can restrict attention to regular broadcast strategies.

\begin{lemma}\label{lem-shape-algo-b}
If there exists a broadcast strategy for a configuration
$Pos[1:n]$ with source agent $a_k$, using power at most $P$, then there exists a regular broadcast strategy for the configuration
$Pos[1:n]$ with source agent $a_k$, using power at most $P$.
\end{lemma} 

\begin{proof}
%
Suppose that there exists a \broadcast strategy for
$(Pos[1:n],k,P)$.  For every agent $a_i$, $i \neq k$ we define $b_i,
f_i$ as in the definition of a regular broadcast strategy. Note that the agents $\{a_i \mid 1 < i < k\}$
execute a reverse pull strategy between $b_{k-1}$ and
$Pos[1]+P$. Similarly, the agents $\{a_i \mid k < i < n\}$ execute a
pull strategy between $b_{k+1}$ and $Pos[n]-P$.  By
Remark~\ref{rem-pull-push}, it means that there exists $i > k$
(resp. $i < k$) such that $a_i$ reaches $b_{k-1}$ (resp. $b_{k+1}$)
with the information from $a_k$.  Moreover, since the agents execute either a reverse pull 
strategy or a pull strategy, we have $Pos[k-1] \leq b_{k-1}$, and $Pos[k+1]\geq
b_{k+1}$.

Suppose the lemma does not hold. This means that $2b_{k+1} - Pos[k] -
b_{k-1} > P$, and $b_{k+1} + Pos[k] - 2b_{k-1} > P$. Consequently,
$a_k$ cannot reach both $b_{k-1}$ and $b_{k+1}$, i.e., there exists $i
<k$ such that $a_i$ reaches $b_{k+1}$, or there exists $i > k$ such
that $a_i$ reaches $b_{k-1}$. If $Pos[k] \leq b_{k-1}$, it implies
that $b_{k+1} > Pos[k] + P$, and consequently, there cannot exist a
{\broadcast} strategy since there is no \subproblem strategy on
$(Pos[k:n-1],Pos[k],Pos[n]-P,P)$. Consequently, we can assume that
$Pos[k] > b_{k-1}$. Using a similar argument we can also assume that $Pos[k] < b_{k+1}$. 


Among all {\broadcast} strategies, consider the strategy that
minimizes the size of $A = \{a_i \mid i < k \mbox{ and } a_i \mbox{
  reaches } b_{k+1}\} \cup \{a_i \mid i > k \mbox{ and } a_i \mbox{
  reaches } b_{k-1}\}$. Without loss of generality, assume that $a_k$
does not reach $b_{k-1}$, and let $i > k$ such that $a_i$ reaches
$b_{k-1}$. For each agent $a_j$, let $x_j, l_j, r_j$ be defined as in
Lemma~\ref{lem-strat1-broadcast}. Note that $r_k \leq Pos[k]+P$ and
$r_i \leq l_i + P \leq b_{k-1} + P \leq Pos[k]+P$. Moreover, $Pos[i] -
P \leq l_i \leq b_{k-1} \leq l_k$.

Consider the new strategy defined as follows: for each agent $j \notin
\{i,k\}$, let $x'_j = x_j, l_j'= l_j$ and $r_j = r_j'$; let $x'_k=x_k =
Pos[k]$, $r_k' = (Pos[k]+Pos[i])/2$ and $l_k' = Pos[i]-P$; let $x'_i =
l'_i = (Pos[k]+Pos[i])/2$ and $r'_i = Pos[k] +P$. Note that
$r'_i+Pos[i]-2l_i' \leq P$ and $2r'_k-Pos[k]-l'_k \leq P$.  Since
$[l_i,r_i] \cup [l_k,r_k] \subseteq [Pos[i]-P,Pos[k]+P] = [l'_k,r'_k]
\cup [l'_i,r'_i]$, this is still a {\broadcast} strategy, in view of
Lemma~\ref{lem-strat1-broadcast}. However, in this new strategy, there
is one agent less in $A' = \{a_i \mid i < k \mbox{ and } a_i \mbox{
  reaches } b_{k+1}\} \cup \{a_i \mid i > k \mbox{ and } a_i \mbox{
  reaches } b_{k-1}\}$ than in $A$, contradicting the choice of our
strategy. 

Consequently, either $2b_{k+1} - Pos[k] - b_{k-1} > P$, or $b_{k+1} +
Pos[k] - 2b_{k-1} > P$ and the lemma holds. 
\end{proof}

\section{{\Cccast} and broadcast on lines}\label{s:line}

\subsection{{\Cccast} on lines}\label{s:line}

In this section we consider the centralized {\convergecast} problem
for lines. We give an optimal, linear-time, deterministic centralized
algorithm, computing the optimal amount of power needed to solve
{\convergecast} for line networks and we provide a regular convergecast strategy for this amount of power. As the algorithm is quite involved,
we start by observing some properties of the optimal
strategies. 

%

\subsubsection{Properties of a {\convergecast} strategy}\label{sec-properties}

In the following, we only consider regular convergecast strategies. Note that a
regular convergecast strategy is fully determined by the value of $P$ and by the
partition of the agents into the two sets $LR$ and $RL$. For each
agent $a_i \in LR$ (resp.  $a_i \in RL$), we denote $f_i$ by
$\rclr(i,P)$ (resp. $\rcrl(i,P)$). Observe that
$\rclr(i,P)$ is the rightmost point on the line to which the set
of $i$ agents at initial positions $Pos[1:i]$, each having power $P$,
may transport their total information. Similarly,
$\rcrl(i,P)$ is the leftmost such point for agents at positions
$Pos[i:n]$.
 
Lemma~\ref{lem:regconv} permits to construct a linear-time decision
procedure verifying if a given amount $P$ of battery power is
sufficient to design a convergecast strategy for a given configuration
$Pos[1:n]$ of agents.  We first compute two lists $\rclr(i,P)$,
for $1 \leq i \leq n$ and $\rcrl(i,P)$, for $1 \leq i \leq
n$. Then we scan them to determine if there exists an index $j$, such
that $\rclr(j,P) \geq \rcrl(j+1,P)$. In such a case, we set
$LR = \{a_r \mid r \leq j\}$ and $RL = \{a_r\mid r> j\}$ and we apply
Lemma~\ref{lem:regconv} to obtain a regular \convergecast strategy where
agents $a_j$ and $a_{j+1}$ meet and exchange their information which
at this time is the entire initial information of the set of agents. If
there is no such index $j$, no {\convergecast} strategy is
possible. This implies

\begin{corollary}
In $O(n)$ time we can decide if a configuration of $n$ agents on the
line, each having a given maximal power $P$, can perform
{\convergecast}.
\end{corollary}

The remaining lemmas of this subsection bring up observations needed 
to construct an algorithm finding the optimal power $P_{OPT}^c$ and designing an optimal {\convergecast} strategy.

Note that if the agents are not given enough power, then it can happen
that some agent $a_p$ may never learn the information from $a_1$ (resp.
from $a_n$). In this case, $a_p$ cannot belong to $LR$
(resp. $RL$). We denote by $\aclr(p)$ the minimum amount of power
needed to ensure that $a_p$ can learn the
information from $a_1$: if $p> 0$, $\aclr(p) = \min \{P\mid
\rclr(p-1,P)+P \geq Pos[p]\}$. Similarly, we have $\acrl(p)
= \min \{P \mid \rcrl(p+1,P)-P \leq Pos[p]\}$.

Given a strategy using power $P$, for each agent $p \in LR$, we have
$P \geq \aclr(p)$ and either $\rclr(p-1,P) \geq Pos[p]$, or
$\rclr(p-1,P) \leq Pos[p]$. In the first case, $\rclr(p,P)
= Pos[p]+P$, while in the second case, $\rclr(p,P)=
2\rclr(p-1,P) +P - Pos[p]$. 

We define threshold functions $\tclr(p)$ and $\tcrl(p)$ that compute, for
each index $p$, the minimal amount of power ensuring that agent $a_p$ 
does not go back when $a_p \in LR$
(respectively $a_p \in RL$), i.e., such that $\rclr(p-1,P)=Pos[p]$ (respectively $\rcrl(p+1,P)=Pos[p]$).  
For each $p$, let $\tclr(p) = \min
\{P \mid \rclr(p,P) = Pos[p]+P\}$ and $\tcrl(p) = \min \{P \mid
\rcrl(p,P) = Pos[p]-P\}$. Clearly, $\tclr(1) = \tcrl(n) = 0$.

The next lemma shows how to compute $\rclr(q,P)$ and $\rcrl(q,P)$
if we know $\tclr(p)$ and $\tcrl(p)$ for every agent $p$.

\begin{lemma}\label{lem-eqn-reach}
Consider an amount of power $P$ and an index $q$. If $p = \max \{p'
\leq q \mid \tclr(p') < P\}$, then $\rclr(q,P) = 2^{q-p}Pos[p]
+ (2^{q-p+1}-1)P - \sum_{i=p+1}^{q} 2^{q-i}Pos[i]$. Similarly, if $p =
\min \{p' \geq q \mid \tcrl(p') < P\}$, then $\rcrl(q,P) =
2^{p-q}Pos[p] - (2^{p-q+1}-1)P - \sum_{i=q}^{p-1} 2^{i-q}Pos[i]$.
\end{lemma}


\begin{proof}
We prove the first statement of the lemma; the proof of the other statement is similar.  We
first show the following claim.
 
\noindent \textbf{Claim.}
If for every $i \in [p+1,q]$, $P
\leq \tclr(i)$, then $$\rclr(q,P) = 2^{q-p}\rclr(p,P) +
(2^{q-p}-1)P - \sum_{i=p+1}^{q} 2^{q-i}Pos[i].$$

We prove the claim by induction on $q-p$.
Note that since $P \leq \tclr(q) $, $\rclr(q,P) =
2\rclr(q-1,P)+P-Pos[q]$.  Thus if $q = p+1$, the statement holds.
Suppose now that $q > p+1$.  Since $q-1 > p$, by the induction hypothesis,
we have

$$\rclr(q-1,P) = 2^{q-1-p}\rclr(p,P) + (2^{q-1-p}-1)P -
\sum_{i=p+1}^{q-1} 2^{q-1-i}Pos[i].$$  

Consequently, we have 
\begin{eqnarray*}
\rclr(q,P) &=& 2\rclr(q-1,P)+P-Pos[q]\\
&=&2^{q-p}\rclr(p,P) + (2^{q-p}-2)P - \sum_{i=p+1}^{q-1} 2^{q-i}Pos[i]+P-Pos[q] \\
&=&2^{q-p}\rclr(p,P) + (2^{q-p}-1)P -
\sum_{i=p+1}^{q} 2^{q-i}Pos[i].\\
\end{eqnarray*}
This concludes the proof of the claim.

If $p = \max \{p' \leq q \mid
\tclr(p') < P\}$, then for each $p' \in [p+1,q]$, $\tclr(p') \geq
P$ and $\rclr(p,P) = Pos[p]+P$. Consequently, 

$$\rclr(q,P) = 2^{q-p}Pos[p] +
(2^{q-p+1}-1)P - \sum_{i=p+1}^{q} 2^{q-i}Pos[i].$$
\end{proof}





In the following, we denote $\sclr(p,q) = \sum_{i=p+1}^{q}
2^{q-i}Pos[i]$ and $\scrl(p,q) = \sum_{i=q}^{p-1} 2^{i-q}Pos[i]$.

\begin{remark}\label{rem-SLR}
For every $p \leq q \leq r$, we have $\sclr(p,r) = 2^{r-q}\sclr(p,q)+\sclr(q,r)$.
\end{remark}

We now show that for an optimal {\convergecast} strategy, the last agent
of $LR$ and the first agent of $RL$ meet at some point between their initial
positions and that they need to use all the available power  to
meet. 

\begin{lemma}\label{lem-egalite-reach}
Suppose there exists an optimal {\convergecast} strategy for a
configuration $Pos[1:n]$, where the maximum power used by an agent is
$P$.  Then, there exists an integer $1 \leq p < n$ such that $Pos[p] <
\rclr(p,P) = \rcrl(p+1,P) < Pos[p+1].$

Moreover, $\forall q \leq p$, $\aclr(q) < P < \tcrl(q)$ and
$\forall q > p$, $\acrl(q) < P < \tclr(q)$.
\end{lemma}

\begin{proof}
In the proof we need the following claim. 

\medskip

\noindent\textbf{Claim. }
For every $1 \leq p \leq n$, the function $\rclr(p,\cdot)$ which assigns the value $\rclr(p,P)$ for any argument $P$, is an increasing, continuous, piecewise linear function with at most $p$ pieces on $[\aclr(p),+\infty)$.

For every $1 \leq p \leq n$, the function $\rcrl(p,\cdot)$ which assigns the value $\rcrl(p,P)$ for any argument $P$, is a decreasing continuous piecewise linear function with at most $p$ pieces on $[\acrl(p),+\infty)$.

\medskip

We prove the first statement of the claim by induction on $p$. For
$p=1$, $\rclr(1,P) = Pos[1]+ P$ and the claim holds. Suppose that
$\rclr(p,\cdot)$ is a continuous piecewise linear function on
$[\aclr(p),+\infty)$ and consider $\rclr(p+1,\cdot)$.

First note that $\aclr(p) < \aclr(p+1)$. Since
$\rclr(p,\cdot)$ is a continuous, increasing function, there
exists a unique $P = \aclr(p+1)$ such that $\rclr(p,P) + P =
Pos[p+1]$ and for every $P' > \aclr(p+1)$, $\rclr(p,P') + P' >
Pos[p+1]$. Consequently, $\rclr(p+1,\cdot)$ is well defined
on $[\aclr(p+1),+\infty)$. 

Since $\rclr(p,\cdot)$ is a continuous, increasing function, there
exists a unique $P = \tclr(p+1)$ such that $\rclr(p,P) =
Pos[p+1]$. If $\aclr(p+1) \geq P \geq \tclr(p+1)$,
$\rclr(p+1,P) = 2\rclr(p,P)+ P - Pos[p+1]$ and thus
$\rclr(p+1,\cdot)$ is an increasing, continuous, piecewise linear
function on $[\aclr(p+1), \tclr(p+1)]$ with at most $p$ pieces. If $P \geq \tclr(p+1)$,
$\rclr(P) = Pos[p+1] + P$ and thus, $\rclr(p+1,\cdot)$ is an
increasing, continuous, linear function on $[\tclr(p+1),
  +\infty)$. Since $ 2\rclr(p,\tclr(p+1)) + \tclr(p+1) -
  Pos[p+1] = Pos[p+1]+ \tclr(p+1)$, the function $\rclr(p+1,\cdot)$ is an
  increasing, continuous, piecewise linear function on $[\aclr(p+1),
    +\infty)$ with at most $p+1$ pieces. 

One can show the second statement of the claim using similar
arguments. This ends the proof of the claim.

Suppose we are given $p$ and consider the partition of the agents into
$LR= \{a_q \mid q \leq p\}$ and $RL = \{a_q \mid q > p\}$. Consider a regular
{\convergecast} strategy for this partition and where the maximum
amount of power $P$ used by an agent is minimized. We first show that
$\rclr(p,P) = \rcrl(p+1,P)$.

Let $Q = \max \{\aclr(p),\acrl(p+1)\}$.  Since
$\rclr(p,\cdot)$ is an increasing continuous function on
$[\aclr(p),+\infty)$ and $\rcrl(p+1,\cdot)$ is a decreasing
continuous function on $[\acrl(p+1),+\infty)$, the difference
$\rclr(p,\cdot) - \rcrl(p+1,\cdot)$ is a continuous
increasing function on $[Q,+\infty)$.

Consider the case where $Q = \acrl(p+1) \geq \aclr(p)$ (the
other case is similar). Since $\rcrl(p+1,Q) = \rcrl(p+2,Q) =
Pos[p+1]+Q$, $\rclr(p,Q) \leq Pos[p]+Q < Pos[p+1]+Q =
\rcrl(p+1,Q)$ and thus, $\rclr(p,Q) - \rcrl(p+1,Q) <
0$.  By definition of a regular convergecast strategy, there exists $Q'$ such that
$\rclr(p,Q') - \rcrl(p+1,Q') \geq 0$. Consequently, since the difference $\rclr(p,\cdot) - \rcrl(p+1,\cdot)$ is a continuous
increasing function on $[Q,+\infty)$, there
exists a unique $Q < P \leq Q'$ such that $\rclr(p,P) =
\rcrl(p+1,P)$.

Consider an optimal regular {\convergecast} strategy and let $P$ be the
maximum amount of power used by any agent. By definition of a regular convergecast strategy,
there
exists an index $p$ such that $\rclr(p,P) = \rcrl(p+1,P)$. 

Suppose that $\rclr(p,P) \leq Pos[p]$. In this case, we have
$\rcrl(p,P) = Pos[p] - P < \rclr(p-1,P)$ since $P >
\aclr(p)$. Consequently, according to  what we have shown above, there
exists $P' < P$ such that $\rcrl(p,P') \leq \rclr(p-1,P')$ and
$P$ is not the optimal value needed to solve {\convergecast}.
This contradiction shows that $Pos[p] <
\rclr(p,P)$.

For similar reasons, if $\rcrl(p+1,P) \geq Pos[p+1]$, $P$ is not the
optimal value needed to solve {\convergecast}. This contradiction shows that
$\rcrl(p+1,P) > Pos[p+1]$.

\medskip
We now prove that for each $q \in [1,p]$, $\aclr(q) < P$. This
follows from the fact that for each $a_q \in LR$ such that $q > 1$, we
have $\aclr(q) > \aclr(q-1)$. Consequently, for each $q \in
[1,p-1]$, $\aclr(q) > \aclr(p)$. Moreover, if $\rclr(p,P)$
is defined, then $P \geq \aclr(p)$. If $P = \aclr(p)$,
then $\rclr(p,P) = Pos[p]-P$ and thus, $\rcrl(p+1,P) \geq
Pos[p+1] - P > Pos[p] - P \geq \rclr(p,P)$. This contradicts the
first statement of the lemma. Hence, we have $P > \aclr(p)$.

For similar reasons, for each $q \in [p+1,n]$, $\acrl(q) < P$. 

\medskip
We finally prove that for each $q \in [1,p]$, $P < TH_{RL}(q)$. Suppose
there exists $q$ such that $P \geq TH_{RL}(q)$ and consider $LR = \{a_r \mid r \leq q-1\}$
and $RL = \{a_r \mid r \geq q\}$. Since $P > \aclr(q)$,
$\rclr(q-1,P) > Pos[q] - P = \rcrl(q,P)$ and consequently,
the first statement of the lemma implies that  there exists $P'
< P$ such that $\rclr(q-1,P') > \rcrl(q,P')$. This
implies that $P$ is not the optimal value needed to solve convergecast. 
This contradiction implies that for each $q \in [1,p]$, $P < TH_{RL}(q)$.

For similar reasons, for each $q \in [p+1,n]$, $P < TH_{LR}(q)$. 
\end{proof}

\subsubsection{A linear algorithm to compute the optimal power needed for
  {\convergecast}}  

We first sketch a suboptimal but much easier algorithm and later present and analyze in detail
a more involved linear-time solution to our problem.
First, we need to compute the functions $\rclr(p,\cdot)$ and $\rcrl(p,\cdot)$ for all $p$ such that $1\leq p\leq n$. By Lemma~\ref{lem-eqn-reach}, the function $\rclr(p,\cdot)$ can be computed from the values $\tclr(q)$ for all $q$ such that $1 \leq q\leq p$. Starting from $p=1$, one can compute all these functions $\rclr(p,\cdot)$, since each value $\tclr(p)=\min\{P\mid \rclr(p-1,P)=Pos[p]\}$ can be deduced from 
$\rclr(p-1,\cdot)$. The computation at step $p$ has a time complexity in $O(p)$ and so the computation of all the functions $\rclr(p,\cdot)$ takes time $O(n^2)$. Similarly, it is possible to compute all the functions $\rcrl(p,\cdot)$, for all $p$ such that $1 \leq p\leq n$, in time $O(n^2)$.  Since $\rclr(p,\cdot)$ and $\rcrl(p+1,\cdot)$ are increasing, continuous, piecewise linear functions with at most $n$ pieces, by the claim from the proof of Lemma~\ref{lem-egalite-reach}, it is possible to compute the value $P$ such that $\rclr(p,P) = \rcrl(p+1,P)$ in time $O(n)$. Hence the optimal value of power needed to achieve {\convergecast} on lines, which is $\min_{1\leq p\leq n}\{P\mid \rclr(p,P) = \rcrl(p+1,P)\}$ by Lemma~\ref{lem-egalite-reach}, can be computed in time $O(n^2)$. 

\medskip

The following result shows that the optimal power needed for convergecast on the line can in fact be computed in linear time.

\begin{theorem}\label{thm:OptPower}
In $O(n)$ time it is possible to compute  the optimal power needed to achieve
{\convergecast} on the line for configuration $Pos[1:n]$  and to compute the optimal convergecast strategy.
\end{theorem}

We first explain how to compute a stack of couples $(p,\tclr(p))$
that we can subsequently use to calculate $\rclr(p,P)$ for any
given $P$.
Then, we present a linear algorithm that computes the value needed to
solve {\convergecast} when the last index $r \in LR$ is provided:
given an index $r$, we compute the optimal power needed to solve
{\convergecast} assuming that $LR = \{a_q\mid q \leq r\}$ and $RL
=\{a_q\mid q > r\}$.
Finally, we explain how to use techniques introduced in the two
previous algorithms in order to compute the optimal power needed to
solve {\convergecast}. 

\paragraph{Computing the threshold values.}

In order to describe explicitly the function $\rclr(q,\cdot)$,
we need to identify the indices $p$ such that for every $r\in
[p+1,q]$, we have $\tclr(r) > \tclr(p)$. They correspond to the
breakpoints at which the slopes of the piecewise linear function
$\rclr(q,\cdot)$ change. Indeed, if we are given such an index
$p$, then for every $P$ comprised between $\tclr(p)$ and $\min
\{\tclr(r) \mid p< r \leq q \}$, we have $\rclr(q,P) =
2^{q-p}Pos[p] + (2^{q-p+1}-1)P - \sclr(p,q)$. We denote by
$X_{LR}(q)$ this set of indices $\{p \leq q \mid \forall r \in
[p+1,q] , \tclr(r) > \tclr(p)\}$.

In particular, if we want to compute $\tclr(q+1)$, we just need to
find $p = \max\{r \leq q \mid \rclr(q, \tclr(r)) < Pos[q+1]\}$,
and then $\tclr(q+1)$ is the value of power $P$ such that
$2^{q-p}Pos[p] + (2^{q-p+1}-1)P - \sclr(p,q) = Pos[q+1]$.  Moreover,
by the choice of $p$, we have $X_{LR}(q+1) = \{r \in X_{LR}(q) \mid r \leq p\}
\cup \{q+1\}$.

Using these remarks, the function \CompThLR, with an input index $r$ of an agent, returns a stack $\mathtt{\tclr}$ containing
couples $(p,P)$ such that $p \in X_{LR}(r)$ and $P= \tclr(p)$.  
Note that in the stack $\mathtt{\tclr}$, the elements $(p,P)$ are sorted
along both components, the largest being on the top of the stack.

The function is described as follows.  Initially, the stack
$\mathtt{\tclr}$ contains only the couple $(1, \tclr(1))$.  At
each iteration, given the stack corresponding to the index $q$, in
order to compute the stack for the index $q+1$, we first pop out all
elements $(p,P)$ such that $\rclr(q,P) > Pos[q+1]$. After that,
the integer $p$ needed to compute $\tclr(q+1)$ is located on the top
of the stack.  Finally, the couple $(q+1,\tclr(q+1))$ is pushed on
the stack before we proceed with the subsequent index $q.$
 The function returns the stack $\mathtt{\tclr}$ corresponding
to the index $r$.

Below, we give the pseudo-code of the function.

\begin{function}[H]
\caption{ThresholdLR(array {$Pos[1:n]$} of real; $r$:integer):stack}\label{algo:cmptTH-s}  
\SetKwData{MyTrue}{true} \SetKwData{MyFasle}{false} 
\SetKw{Not}{not} \SetKw{MyAnd}{and}
\SetKw{Pop}{pop}
\SetKw{Push}{push}
\SetKw{KwDownTo}{down to}
\SetKw{EmptyStack}{empty\_stack}

$\mathtt{\tclr} = \EmptyStack$\;
\Push($\mathtt{\tclr}$,$(1,0)$)\;
\For{$q=1$ \KwTo $r-1$}
{
  $(p,P) = \Pop(\mathtt{\tclr})$\tcc*{$p = q$ and $P = \tclr(p)$}
  \While{$2^{q-p}*Pos[p] + (2^{q-p+1}-1)*P-\sclr(p,q) \geq Pos[q+1]$} 
        {$(p,P) = \Pop(\mathtt{\tclr})$\;}
  \tcp{while $\rclr(q,P) \geq Pos[q+1]$ we consider the next
    element in $\mathtt{\tclr}$}
  \Push($\mathtt{\tclr}$,$(p,P)$)\;
  $Q = (2^{q-p}*Pos[p] -Pos[q+1] - \sclr(p,q))/(2^{q-p+1}-1)$\;
  \tcc{$Q$ is the solution of $\rclr(q,P) = Pos[q+1]$}
  \Push($\mathtt{\tclr}$,$(q+1,Q)$)\;
}
\Return($\mathtt{\tclr}$)\;
\end{function}

The number of stack operations performed during the execution of this
function is $O(r)$. However, in order to obtain a linear number of
arithmetic operations, we need to be able to compute $2^{q-p}$ and
$\sclr(p,q)$ in constant time.

In order to compute $2^{q-p}$ efficiently, we can store the values of
$2^i$, $i \in [1,n-1]$ in an auxiliary array, that we have precomputed
in $O(n)$ time. We cannot precompute all values of $\sclr(p,q)$ since
this requires calculating $\Theta(n^2)$ values. However, from
Remark~\ref{rem-SLR}, we know that $\sclr(p,q) = \sclr(1,q) -
2^{q-p}\sclr(1,p)$. Consequently, it is enough to precompute
$\sclr(1,i)$ for each $i \in [2,n]$. Since $\sclr(1,i+1) =
2\sclr(1,i)+Pos[i+1]$, this can be done using $O(n)$ arithmetic
operations.

Similarly, we can define the function \CompThRL(\texttt{array}
$Pos[1:n]$ \texttt{of real}, $r$\texttt{:integer):stack} that returns
a stack $\mathtt{\tcrl}$ containing all pairs $(q,\tcrl(q))$ such
that for every $p \in [r,q-1]$, we have $\tcrl(p)> \tcrl(q)$.

\paragraph{Computing the optimal power when $LR$ and $RL$ are known.} 

To facilitate further reading, we first show how to compute the optimal power $P_{OPT}^c$,
if the sets $LR$ and $RL$ are known. This will be done by function \CompOptimalPos which will be not used in our final algorithm to compute optimal power but whose role is to explain some of the techniques under these additional assumptions.

Suppose that we are given an agent index $r$ and we want to
compute the optimal power needed to solve {\convergecast} when $LR =
\{a_p \mid p \leq r\}$ and $RL = \{a_q \mid q > r\}$. From
Lemma~\ref{lem-egalite-reach}, we know that there exists a unique value
$P_{OPT}^c$ such that $\rclr(r,P_{OPT}^c) = \rcrl(r+1,P_{OPT}^c)$.

As previously, by Lemma~\ref{lem-eqn-reach}, we know that the value
of $\rclr(r,P_{OPT}^c)$ depends on $p = \max\{p' \leq r \mid
\tclr(p') < P_{OPT}^c\}$. Similarly, $\rcrl(r+1,P_{OPT}^c)$ depends
on $q = \max\{q' \geq r+1 \mid \tcrl(q') < P_{OPT}^c\}$.  If we are
given the values of $p$ and $q$, then $P_{OPT}^c$ is the unique value of $P$
such that
$$ 2^{r-p}Pos[p] - (2^{r-p+1}-1)P - \sclr(p,r) = 2^{q-r-1}Pos[q] -
          (2^{q-r}-1)P - \scrl(q,r+1).$$ 

In Function \CompOptimalPos, we first use functions \CompThLR and \CompThRL to
compute the two stacks $\mathtt{\tclr}$ and $\mathtt{\tcrl}$
containing respectively $\{(p,\tclr(p)) \mid p \in X_{LR}(r)\}$ and
$\{(q,\tcrl(q)) \mid q \in X_{RL}(r+1)\}$.  Then at each iteration,
we consider the two elements $(p,P_{LR})$ and $(q, P_{RL})$ that are
on top of both stacks. If $P_{LR} \geq P_{RL}$ (the other case is
symmetric), we check whether $\rclr(r,P_{LR}) \geq
\rcrl(r+1,P_{LR})$.  In this case, we have $P > P_{OPT}^c$, so we
remove $(p,P_{LR})$ from the stack $\mathtt{\tclr}$ and we proceed
to the next iteration.  If $\rclr(r,P_{LR}) <
\rcrl(r+1,P_{LR})$, we know that $P_{OPT}^c \geq P_{LR} \geq
P_{RL}$ and we can compute the value of $P_{OPT}^c$ using
Lemma~\ref{lem-eqn-reach}.

Let $Y_{LR}(r,P)$ denote $ \{(p,\tclr(p)) \mid p \in X_{LR}(r)
\mbox{ and } \tclr(p) < P\}$ and $Y_{RL}(r+1,P) = \{(q,\tcrl(q))
\mid q \in X_{RL}(r+1) \mbox{ and } \tcrl(q) < P\}$.

\begin{remark}\label{rem-atpos}
At the end of the execution of Function \CompOptimalPos,
$\mathtt{\tclr}$ and $\mathtt{\tcrl}$ contain respectively
$Y_{LR}(r,P_{OPT}^c)$ and $Y_{RL}(r+1,P_{OPT}^c)$.

Moreover, if initially the two stacks $\mathtt{\tclr}$ and
$\mathtt{\tcrl}$ contain respectively $Y_{LR}(r,P)$ and
$Y_{RL}(r+1,P)$ for some $P \geq P_{OPT}^c$, then the value computed by
the function is also $P_{OPT}^c$ .
\end{remark}

The pseudo-code of the of Function \CompOptimalPos is given below.

\begin{function}[h!]
\caption{OptimalAtIndex (array {$Pos[1:n]$} of real;
  r:integer):stack \label{algo:cmptTH}}   
\SetKwData{MyTrue}{true} \SetKwData{MyFasle}{false} 
\SetKw{Not}{not} \SetKw{MyAnd}{and}
\SetKw{Pop}{pop}
\SetKw{Push}{push}
\SetKw{KwDownTo}{down to}
\SetKw{EmptyStack}{empty\_stack}

$\mathtt{\tclr} = \CompThLR(r)$; $\mathtt{\tcrl} =
\CompThRL(r+1)$\;

$(p,P_{LR}) = \Pop (\mathtt{\tclr})$; $(q,P_{RL}) = \Pop (\mathtt{\tcrl})$;
$P = \max \{P_{LR},P_{RL}\}$\;
\tcc{$p = r$, $P_{LR}={\tclr}(r)$, $q = r+1$, $P_{RL}={\tcrl}(r+1)$.}

\While(\\ \tcc*[h]{While $\rclr(r,P) \geq
          \rcrl(r+1,P)$ do}){$2^{r-p}Pos[p] + (2^{r-p+1}-1)P - \sclr(p,r)
\geq  2^{q-r-1}Pos[q] - (2^{q-r}-1)P - \scrl(q,r+1) 
$}  
      {
        \lIf{$P_{LR} \geq P_{RL}$}
            {
              $(p,P_{LR}) = \Pop(\mathtt{\tclr})$}
        \lElse{
              $(q,P_{RL}) = \Pop(\mathtt{\tcrl})$}
        $P =  \max \{P_{LR},P_{RL}\}$\;    
      }      
      $P_{OPT}^c = (2^{q-r-1}Pos[q] - \scrl(q,r+1) - 2^{r-p}Pos[p]+\sclr(p,r))/(2^{r-p+1}+2^{q-r}-2)$\; 
\tcc{$P_{OPT}^c$ is the solution of $\rclr(r,P_{OPT}^c) =
  \rcrl(r+1,P_{OPT}^c)$}  
\Return($P_{OPT}^c$)\; 
\end{function}

\paragraph{Computing the optimal power for \convergecast.}

We now explain how to compute the optimal amount of power needed to
achieve {\convergecast} using a linear number of operations. Notice that
Function \CompOptimalPos does it only provided the partition of the agents
in $LR$ and $RL$.

Let $P_{<r}$ be the optimal value needed to solve convergecast when
$\max\{s\mid a_{s} \in LR\} < r$, i.e., when the two agents whose
meeting results in merging the entire information are $a_i$ and $a_{i+1}$
for some $i < r$. If $\rclr(r,P_{<r}) \leq \rcrl(r+1,P_{<r})$, then
$P_{<r+1} = P_{<r}$. However, if $\rclr(r,P_{<r}) >
\rcrl(r+1,P_{<r})$, then $P_{<r+1} < P_{<r}$ and $P_{<r+1}$ is
the unique value of $P$ such that $\rclr(r,P) =
\rcrl(r+1,P)$. This corresponds to the value returned by function
\CompOptimalPos($Pos,r$).

The general idea of Function \CompOptimal is to iteratively compute
the value of $P_{<r}$. If we need a linear time algorithm, we cannot
call repeatedly the function \CompOptimalPos. However, from
Remark~\ref{rem-atpos}, in order to compute $P_{<r+1}$ when $P_{<r+1}
\leq P_{<r}$, it is enough to know $Y_{LR}(r,P_{<r})$ and
$Y_{RL}(r+1,P_{<r})$.  If we know $Y_{LR}(r,P_{<r})$ and
$Y_{RL}(r+1,P_{<r})$, then we can use the same algorithm as in
\CompOptimalPos in order to compute $P_{<r+1}$. Moreover, from
Remark~\ref{rem-atpos}, we also get $Y_{LR}(r,P_{<r+1})$ and
$Y_{RL}(r+1,P_{<r+1})$ when we compute $P_{<r+1}$.

Before proceeding to the next iteration, we need to compute
$Y_{LR}(r+1,P_{<r+1})$ and $Y_{RL}(r+2,P_{<r+1})$ from
$Y_{LR}(r,P_{<r+1})$ and $Y_{RL}(r+1,P_{<r+1})$. Note that if
$\tclr(r) > P_{<r+1}$, then $Y_{LR}(r+1,P_{<r+1}) =
Y_{LR}(r,P_{<r+1})$. If $\tclr(r) \leq P_{<r+1}$, we can use the
same function as in \CompThLR to compute $Y_{LR}(r+1,P_{<r+1}) =
\{(p,\tclr(p)) \mid p \in X_{LR}(r)\}$ from $Y_{LR}(r,P_{<r+1})$.
Consider now $Y_{RL}(r+2,P_{<r+1})$. If $\tcrl(r+1) > P_{<r+1}$,
then $(r+1,\tcrl(r+1)) \notin Y_{RL}(r+1,P_{<r+1})$, and
$Y_{RL}(r+2,P_{<r+1}) = Y_{RL}(r+1,P_{<r+1})$. If $\tcrl(r+1) \leq
P_{<r+1}$, then either $Pos[r+1]-P_{<r+1} \geq
\rcrl(r+1,P_{<r+1})$ if $P_{<r+1} = P_{<r}$, or 
$Pos[r+1]-P_{<r+1} = \rcrl(r+1,P_{<r+1}) =
\rclr(r,P_{<r+1})$ if $P_{<r+1} < P_{<r}$. In both cases, it
implies that $\aclr(r+1) \geq P_{<r+1}$.  Therefore,
by Lemma~\ref{lem-egalite-reach}, $P_{<i} = P_{<r+1}$ for every $i \geq
r+1$ and we can return the value of $P_{<r+1}$.

In Function \CompOptimal, at each iteration, the stack
$\mathtt{\tclr}$ contains $Y_{LR}(r,P_{<r})$ (except its top
element) and the stack $\mathtt{\tcrl}$ contains
$Y_{RL}(r+1,P_{<r})$ (except its top element). Initially,
$\mathtt{\tclr}$ is empty and $\mathtt{\tcrl}$ contains $O(n)$
elements. In each iteration, at most one element is pushed into the
stack $\mathtt{\tclr}$ and no element is pushed into the stack
$\mathtt{\tcrl}$. Consequently, the number of stack operations
performed by Function \CompOptimal is linear.

\begin{function}[h]
\caption{ComputeOptimal (array $Pos{[1:n]}$ of real):real\label{algo:opt}}
\SetKwData{MyTrue}{true} \SetKwData{MyFasle}{false} 
\SetKw{Not}{not} \SetKw{MyAnd}{and}
\SetKw{Pop}{pop}
\SetKw{Push}{push}
\SetKw{EmptyStack}{empty\_stack}

$\mathtt{\tclr} = \EmptyStack$; $\mathtt{\tcrl} = \CompThRL(Pos)$\; 

$(q,P_{RL}) = \Pop(\mathtt{\tcrl})$; $P_{OPT}^c = P_{RL}$\; 
\tcc{$q = 1$, $P_{RL} = \tcrl(1)$}
$(q,P_{RL}) = \Pop(\mathtt{\tcrl})$; 
$p=1$;$P_{LR}=0$\;
\For{$r = 1$ \KwTo $n-1$}
{
  \tcc{$P_{OPT}^c = P_{<r} \geq P_{LR}, P_{RL}$
  }
  \If{$2^{r-p}Pos[p] + (2^{r-p+1}-1)P_{OPT}^c - \sclr(p,r) > 
    2^{q-r-1}Pos[q] - (2^{q-r}-1)P_{OPT}^c - \scrl(q,r+1)$} 
     {
       \tcc{If $\rclr(r,P_{OPT}^c) > \rcrl(r+1,P_{OPT}^c)$ then 
         $P_{OPT}^c$ is larger than the value needed to solve
         {\convergecast} at position $r$. We apply now the same algorithm
         as in function \CompOptimalPos. 
       }
        $P =  \max \{P_{LR},P_{RL}\}$\;    
       \While{$2^{r-p}Pos[p] + (2^{r-p+1}-1)P - \sclr(p,r)\geq
           2^{q-r-1}Pos[q] - (2^{q-r}-1)P - \scrl(q,r+1) 
         $}  
             {
               \lIf{$P_{LR} \geq P_{RL}$}
                   {
                     $(p,P_{LR}) = \Pop(\mathtt{\tclr})$}
                   \lElse{
                     $(q,P_{RL}) = \Pop(\mathtt{\tcrl})$}
                   $P =  \max \{P_{LR},P_{RL}\}$\;    
             }      
             $P_{OPT}^c = (2^{q-r-1}Pos[q] - \scrl(q,r+1) -
             2^{r-p}Pos[p]+\sclr(p,r))/(2^{r-p+1}+2^{q-r}-2)$\;  
             \tcc{$P_{OPT}^c = P_{<r+1}$ is the solution of $\rclr(r,P_{OPT}^c) =
               \rcrl(r+1,P_{OPT}^c)$}  
     }

     \lIf{$q = r+1$} 
         {
           \Return $P_{OPT}^c$}
         \tcc{In this case, $P_{OPT}^c \geq \tcrl(r+1)$ and thus
           $P_{OPT}^c = P_{<r} = \aclr(r+1)$: for any $s > r$, $P_{<s} =
           P_{<r}$}
     \If{$2^{r-p}*Pos[p] + (2^{r-p+1}-1)*P_{OPT}^c-\sclr(p,r) \geq Pos[r+1]$}
         {
           \tcc{If $\rclr(r,P_{OPT}^c) \geq Pos[r+1]$ then
             ${\tclr}(r+1) \leq P_{OPT}^c$ and we update
             $\mathtt{\tclr}$, using the same 
             algorithm as in function \CompThLR. }
           \While{$2^{r-p}*Pos[p] + (2^{r-p+1}-1)*P_{LR}-\sclr(p,r)$}
           {$(p,P_{LR}) = \Pop(\mathtt{\tclr})$\;}
           \Push($\mathtt{\tclr}$,($p,P_{LR}$))\;      
           $P_{LR} = (Pos[r+1]+\sclr(p,r)-2^{r-p}*Pos[p])/(2^{r-p+1}-1)$\;
           $p=r+1$\; 
         }     
}
\end{function}

Notice that the partition of agents into sets $LR$ and $RL$ is given by the value of index $r$ when $P_{OPT}^c$ is returned by Function
\CompOptimal.
Since an optimal regular convergecast strategy is fully determined by the value of $P_{OPT}^c$ and by the partition of the agents into the sets $LR$ and $RL$, Function \CompOptimal also yields an optimal convergecast strategy. Hence, this concludes the proof of Theorem~\ref{thm:OptPower}.
\vspace{-0.2cm}

\subsection{{\Cbcast} on lines}\label{s:line-b}

In this section we consider the centralized {\broadcast} problem
for lines. We give an optimal, linear-time, deterministic centralized
algorithm, computing the optimal amount of power needed to solve
{\broadcast} for line networks and computing an optimal broadcast strategy. 



\subsubsection{Properties of a {\broadcast} strategy}\label{sec-properties-b}

In the following, we only consider regular broadcast strategies. Note that a
regular broadcast strategy is fully determined by the value of $P$ and by 
the two possible values of $b_k$ for the source agent $a_k$ ($b_k=b_{k-1}$ or $b_k=b_{k+1}$). 

Let $LR = \{a_1,a_2,\hdots,a_{k-1}\}$ and $RL = \{a_{k+1}, a_{k+2},\hdots,a_n\}$. (Note that we slightly abuse notation by using the same names $LR$ and $RL$ for subsets of agents as in convergecast.)
For each agent $a_i\in LR$  (resp. $a_i\in RL$ ), we denote $b_i$ by
$\rblr(i,P)$ (resp. $\rbrl(i,P)$). Observe that
$\rblr(i,P)$ is the rightmost point on the line from which the set
of $i$ agents at initial positions $Pos[1:i]$, each having power $P$,
may pick the information and bring it back to $a_1$. Similarly,
$\rbrl(i,P)$ is the leftmost point from which the agents at positions
$Pos[i:n]$ may pick the information and bring it back to $a_n$.

Lemma~\ref{lem-shape-algo-b} permits to construct a linear-time decision
procedure verifying if a given amount $P$ of battery power is
sufficient to design a broadcast strategy for a given configuration
$Pos[1:n]$ of agents and a specified source agent $a_k$.  We first compute $b_{k-1}=\rblr(k-1,P)$ and $b_{k+1}=\rbrl(k+1,P)$. 
Then we test if $|2\rblr(k-1,P) - Pos[k] - \rbrl(k+1,P)|$ or  $|2\rbrl(k+1,P) - Pos[k] - \rblr(k-1,P)|$ are less or equal than $P$.
If one of the inequalities is true then there is a broadcast strategy. Otherwise, broadcast is not possible. This implies

\begin{corollary}
In $O(n)$ time we can decide if a configuration $Pos[1:n]$ of $n$ agents on the
line, each having a given maximal power $P$, can perform
{\broadcast} for a given source agent.
\end{corollary}


Note that if the agents are not given enough power, then it can happen
that some agent $a_p,$ $1 \le p\le k$ (resp. $k \le q \le n$) cannot reach the
point $\rblr(p-1,P)$ (resp. $\rbrl(q+1,P)$). We denote by
$\ablr(p)$ (resp. $\abrl(q)$) the minimum amount of power $P$ we have to give the agents
to ensure that $a_p$ (resp. $a_q$) can reach $\rblr(p-1,P)$ (resp. $\rbrl(q+1,P)$). We have : $\ablr(1) = \abrl(n) = 0$ and if $2 \le p\le k,$
$\ablr(p) = \min \{P\mid \rblr(p-1,P) \geq Pos[p]-P\}$. 
Similarly, if $k \le q \le n-1,$ we have $\abrl(q) = \min \{P \mid
\rbrl(q+1,P) \leq Pos[q]+P\}$.

In a regular broadcast strategy using power $P$, for each agent $p \in LR$ such that $P
\geq \ablr(p),$ we have $\rblr(p,P)=(\rblr(p-1,P) + P + Pos[p])/2$.
Similarly, for each agent $q \in RL$ such that $P
\geq \abrl(q),$ we have $\rbrl(q,P)=(\rbrl(q+1,P) - P + Pos[q])/2$.
The next lemma shows how to compute $\rblr(p,\cdot)$ on the interval 
$[\ablr(p),+\infty)$ for every $p\in \{1,2,\hdots,k\}$ and
$\rbrl(q,\cdot)$ on the interval $[\abrl(q),+\infty)$ for every $q\in
\{k,k+1,\hdots,n\}.$

\begin{lemma}\label{lem-eqn-reach-b}
Consider an index $p\in \{1,2,\hdots,k\}$ and an amount of power 
$P\ge \ablr(p),$ then $\rblr(p,P) = \rblr(p,\ablr(p)) + P - \ablr(p).$
Analogously, for an index $q\in \{k,k+1,\hdots,n\},$ and an amount of 
power $P\ge \abrl(q),$ we have $\rbrl(q,P) = \rbrl(q,\abrl(q)) - P + \abrl(q).$
\end{lemma}

\begin{proof}
First, we show by induction on $p$ that for any $p\in \{1,2,\hdots,k\}$ and an amount of power $P\ge \ablr(p)$, we have $\rblr(p,P) = \rblr(p,\ablr(p)) + P - \ablr(p).$
This is true for $p=1$ since $\rblr(1,P) = P$, $\ablr(1)=0$ and $\rblr(1,\ablr(1))=0$.
Now, assume by induction that $\rblr(p-1,P) = \rblr(p-1,\ablr(p-1)) + P - \ablr(p-1)$.
By definition of a regular broadcast strategy, we have for all $P\geq \ablr$:

\[
\begin{aligned}
\rblr(p,P)&=\frac{\rblr(p-1,P)+Pos[p]+P}{2}\\
&=\frac{\rblr(p-1,\ablr(p-1)) + P - \ablr(p-1)+Pos[p]+P}{2}\\
&=P+\frac{\rblr(p-1,\ablr(p-1)) - \ablr(p-1)+Pos[p]}{2}\\
\end{aligned}
\]

Observe that for $P=\ablr(p)$, we have: 
\begin{eqnarray*}
\rblr(p,\ablr(p)) & = &\frac{\rblr(p-1,\ablr(p-1))  - \ablr(p-1)+Pos[p]}{2} + \ablr(p).
\end{eqnarray*}
Hence we have:
$$\rblr(p,P)=  \rblr(p,\ablr(p)) + P - \ablr(p).$$

This concludes the proof by induction.

Similarly, we can show by induction on $q$ that for any $\in \{k,k+1,\hdots,n\}$ and an amount of power $P\ge \ablr(p)$, we have $\rblr(q,P) = \rblr(q,\ablr(q)) - P + \ablr(q).$
%
%
%
\end{proof}

\subsubsection{A linear algorithm to compute the optimal power needed for
  {\broadcast}}  

In this section, we prove the following theorem. 

\begin{theorem}\label{thm:OptPower-b}
In $O(n)$ time it is possible to compute the optimal power needed to achieve broadcast for a configuration $Pos[1 : n]$ of $n$ agents on the line for any source agent and to compute an optimal broadcast strategy.
\end{theorem}

\begin{proof}
We formulate Function \Optbc which computes in linear time the optimal power for the broadcast in the line. 

\begin{function}[h!]
\caption{OptimalBroadcast (array {$Pos[1:n]$} of real;
  r:integer):real \label{algo:optimalBroadcast}}   
\SetKwData{MyTrue}{true} \SetKwData{MyFasle}{false} 
\SetKw{Not}{not} \SetKw{MyAnd}{and} \SetKw{MyOr}{or}
\SetKw{Pop}{pop}
\SetKw{Push}{push}
\SetKw{KwDownTo}{down to}
\SetKw{EmptyStack}{empty\_stack}

$\ablr(1)=0,\rblr(1,\ablr(1))=0, p=1$\;
\While{$p<k-1$}  
      {
        \While{$(p<k-1)$ \MyAnd $(Pos[p+1] - \ablr(p) \le \rblr(p,\ablr(p)))$}
            {
                $\ablr(p+1) = \ablr(p)$\;
                $\rblr(p+1,\ablr(p+1))=(Pos[p+1]+\ablr(p+1)+ \rblr(p,\ablr(p+1)))/2$\;
			$p=p+1$\;
            }
        \If{$(p<k-1)$}
		{
		$\delta_{LR} = (Pos[p+1]-\ablr(p)- \rblr(p,\ablr(p)))/2$\;
		$\ablr(p+1) = \ablr(p)+\delta_{LR}$\;
				$\rblr(p+1, \ablr(p+1)) = Pos[p+1]$\;
		$p=p+1$\;
	        }

	        }

$\abrl(n)=0, \rbrl(n,\abrl(n))=Pos[n], q=n$\;        
       \While{$q>k+1$}  {
                \While{$(q>k+1)$ \MyAnd $(Pos[q-1] + \abrl(q) \ge \rbrl(q,\abrl(q))$}
            {
                          $\abrl(q-1) = \abrl(q)$\;
              $\rbrl(q-1, \abrl(q-1))=(Pos[q-1]- \abrl(q-1)+  \rbrl(q-1,\abrl(q-1))/2$\;
              $q=q-1$\;
            }
        \If{$(q>k+1)$ \MyAnd $(p=k-1$ \MyOr $\delta_{LR}< \delta_{LR})$}{
               	$\delta_{RL} = (\rbrl(q,\abrl(q))-\abrl(q)-Pos[q-1])/2$\;
	        $\abrl(q-1) = \abrl(q)$\;
				$\rbrl(q-1, \abrl(q-1)) = Pos[q-1]$\;
	        		$q=q-1$\;
	      }
        }
       
 $P = \max(\ablr(k-1),\abrl(k+1))$\; 
\If{$\ablr(k-1) > \abrl(k+1))$}{
 $\X=\rblr(k-1,\ablr(k-1))$\;
  $\Y=\rbrl(k+1,\abrl(k+1))+\ablr(k-1) - \abrl(k+1))$\;
 }
 \Else{
  $\X=\rblr(k-1,\ablr(k-1))-\abrl(k+1) + \ablr(k-1))$\;
  $\Y=\rbrl(k+1,\abrl(k+1))$\; 
 }
         
 \If{$Pos[k] \le \X$}
     {$\delta =( \Y-Pos[k]-P)/2$\;}
 \If{$\X < Pos[k] < \Y$}
     {$\delta = (\min(Pos[k]-\X,\Y-Pos[k]) + (\Y-\X) - P)/2$\;}
 \If{$\Y \le Pos[k]$}
     {$\delta = (Pos[k]-\X-P)/2$\;}
 \Return $P+\max(0,\delta)$
\end{function}

In order to compute this value, Function \Optbc first computes the minimal amount of power $Q$ such that all agents in $LR\cup RL$ are activated, i.e., $Q = \max(\ablr(k-1),\abrl(k+1))$. In order to compute $Q$, the function iteratively increases the power sufficient to activate all agents in $LR$. Then, it does the same with agents in $RL$. The function computes iteratively for each agent $a_i$ from $a_1$ to $a_{k-1}$ in $LR$ (respectively from $a_{n}$ to $a_{k+1}$ in $RL$), the value $\ablr(i)$ (respectively $\abrl(i)$) and the value $\rblr(i,\ablr(i))$ (respectively $\rbrl(i,\abrl(i))$). Once $Q$ is known, the function computes the minimal amount of power $P\geq Q$ that enables the agent $a_k$ to reach $\rblr(k-1,P)$ and $\rbrl(k+1,P)$.
This will be proved to be the minimal power to accomplish broadcast.

Notice that in order to accomplish broadcast, agent $a_{k-1}$ must be able to reach $\rblr(k-2,\ablr({k-2}))$. Hence the optimal value $P_{OPT}^b$ of power sufficient to accomplish broadcast  must be at least $\ablr(k-1)$. Similarly, $P_{OPT}^b$ must be at least $\abrl(k+1)$. Hence, we will first prove that the values of $\rblr(p,\ablr(p))$, $\ablr({p})$, $\rbrl(q,\abrl({q}))$ and $\abrl({q})$ are correctly computed for $1\leq p<k$ and $k< q\leq n$. 

We only prove that the values of $\rblr(p,\ablr({p}))$ and $\ablr({p})$ 
are correctly computed for $1\leq p<k$, as the proof that $\rbrl(q,\abrl(q))$ and $\abrl(q)$ are correctly computed for $k< q\leq n$ is similar. The proof is by induction on $p$. For $p=1$, the values of $\rblr(p,\ablr({p}))$ and $\ablr(p)$ 
are correctly computed since $\rblr(1,\ablr({1}))=0$ and $\ablr({1})=0$. Suppose that the values of $\rblr(p,\ablr({p}))$ and $\ablr({p})$ are correctly computed. If $\rblr(p,\ablr(p))\geq Pos[p+1]-P$, then $\ablr(p+1)=\ablr(p)$. By Lemma~\ref{lem-eqn-reach-b}, all functions $\rblr(p,P)$ are linear with coefficient 1 on $[\ablr(p),+\infty)$. Hence, if $\rblr(p,\ablr(p))< Pos[p+1]-\ablr(p)$, we have $\ablr(p+1)=\ablr(p)+(Pos[p+1]-P-\rblr(p,\ablr(p)))/2$. This shows that $\ablr(p+1)$ is correctly computed. It remains to show that $\rblr(p+1,\ablr(p+1))$ is correctly computed. By definition of a regular broadcast strategy, we have 
$\rblr(p+1,\ablr(p+1))=(\rblr(p,\ablr(p+1)) + \ablr(p+1) + Pos[p+1])/2$. If $\rblr(p,\ablr(p))\geq Pos[p+1]-\ablr(p)$, then 
$\rblr(p+1)$ is correctly computed as the above formula is used by the function. Otherwise, we have :
$\ablr(p+1)=\ablr(p)+(Pos[p+1]-\ablr(p)-\rblr(p,\ablr(p)))/2$. Using the notation $r=\rblr(p,\ablr(p))$, $r'=\rblr(p,\ablr(p+1))$,
$a=\ablr(p)$, $a'=\ablr(p+1)$ we have :

\begin{eqnarray*}
\rblr(p+1,\ablr(p+1)) & = & \frac{r' + a' + Pos[p+1]}{2}\\
& = & \frac{(r+a'-a) +(a+(Pos[p+1]-a-r)/2) + Pos[p+1]}{2}\\
& = & \frac{(r+(Pos[p+1]-a-r)/2)) +(a+(Pos[p+1]-a-r)/2) + Pos[p+1]}{2}\\
& = & Pos[p+1].
\end{eqnarray*}

This completes the proof by induction.

Again, using the fact that all functions $\rblr(p,P)$ are linear with coefficient 1 on $[\ablr(p),+\infty)$, the function \Optbc computes correctly the value $\X=\rblr(k-1,Q)$. The same is true for $\Y=\rbrl(k+1,Q)$. Finally, we consider three cases : $Pos[k] \le \X$, $\X < Pos[k] < \Y$ or $\Y \le Pos[k]$ to compute the additional power $\delta$ that has to be used.
By definition of $\rblr(k-1,P)$ and $\rbrl(k+1,P)$, we conclude that $P$ is the optimal value of power to achieve broadcast by a regular strategy. In view of Lemma~\ref{lem-shape-algo-b}, this concludes the proof that $P$ is the optimal value of power to achieve broadcast. The complexity $O(n)$ of the function is straightforward by its formulation. 

Since a regular broadcast strategy is fully determined by the value of $P$ and by 
the two possible values of $b_k$ for the source agent $a_k$ ($b_k=b_{k-1}$ or $b_k=b_{k+1}$), computing the optimal power $P$ yields an optimal broadcast strategy. This concludes the proof of Theorem~\ref{thm:OptPower-b}.
%
%
\end{proof}

\section{{\Cccast} and broadcast on trees  and graphs}\label{s:tree}

We start the section by showing that for arbitrary trees the {\cccast} problem and the centralized  broadcast problem are substantially harder than on lines.

A configuration for {\ccast} on arbitrary graphs is a couple $(G,A)$ where $G$ is a $n$-node weighted graph representing the network and $A$ of size $k$ is the set of the starting nodes of the agents. A configuration for broadcast additionally specifies the starting node of the source agent. 
We consider the {\cccast} decision problem and the broadcast decision problem formalized as follows. 

\probleme{Centralized {\convergecast} decision}{a configuration $(G,A)$ and a real $P$.}{Is there a convergecast {\strat} for $(G,A)$, in which each agent uses at most $P$ battery power~?}

\probleme{Centralized {\broadcast} decision}{a configuration $(G,A)$ with a specified source agent and a real $P$.}{Is there a broadcast {\strat} for $(G,A)$ with the specified source agent, in which each agent uses at most $P$ battery power ?}

We will prove that both these problems are strongly NP-complete.
In order to do this, we consider \emph{star configurations}, i.e., configurations $(G,A)$ in which $G$ is a star, i.e., a tree of diameter 2. We define a class of strategies in a star 
called \emph{simple} that consist of the following two phases :

\begin{itemize}
\item The strategy starts with a gathering phase lasting time $P$,
in which each agent uses all its available power to move towards the center of the star and then waits until time $P$. The agents that have used all their power during this phase without reaching the center are called \emph{depleted}.
\item In the second phase, the agents does not move past depleted agents, i.e., never enter the segment between a leaf and a depleted agent.
\end{itemize}
The following lemma shows that it is enough to consider simple strategies for convergecast and broadcast.

\begin{lemma}\label{cl:simpl}
If there exists a {\ccast} {\strat} (respectively a broadcast strategy) in a star using power $P$, then there exists a simple {\ccast} {\strat} (respectively a simple broadcast strategy) using power $P$.
\end{lemma}

\begin{proof}
Let $\cA$ be a convergecast or a broadcast {\strat}. We construct a simple {\strat} $\cA'$ as follows. In $\cA'$, each agent 
moves towards the center of the star until it has used all its battery power or has reached the center of the star. This gathering phase lasts from time $0$ to time $P$. If an agent has not reached the center in strategy $\cA$, then it stops forever in $\cA'$.
Otherwise, consider time $t$ at which it arrives at the center in $\cA$. Then, in strategy $\cA'$, the agent executes at time $t'+P$ each movement performed at time $t'\geq t$ in {\strat} $\cA$. However, if a movement of an agent would result in the agent moving past a depleted agent from time $r$ to $r'$ in $\cA$, then in strategy $\cA'$ the agent waits at the position of the depleted agent instead of moving past it.  By construction, $\cA'$ is a simple {\strat}. Observe that in {\strat} $\cA'$, the non-depleted agents share all their information at the center of the star at time $P$. 
Since two depleted agents cannot meet, it remains to show that when a non-depleted agent $b$ meets a depleted agent $a$ at time $t$ in {\strat} $\cA$, they meet at time $t+P$ in $\cA'$. The final position of agent $a$ is not farther from the center in $\cA'$ than in $\cA$. Hence, any agent $b$ that meets agent $a$ at time $t$ is at the new position of $a$ in $\cA'$ at time $t+P$. Hence, the meeting between $a$ and $b$ occurs in $\cA'$ as well. If $\cA$ was a convergecast strategy (respectively a broadcast strategy) then $\cA'$ is a simple convergecast strategy (respectively a simple broadcast strategy).
\end{proof}

\begin{theorem}\label{th:NP-graph}
The {\cccast} decision problem and the centralized broadcast decision problem are strongly NP-complete for trees.  
\end{theorem}

The proof of Theorem \ref{th:NP-graph} is split into three lemmas. We first show that the {\cccast} decision problem is strongly NP-hard, then that the centralized broadcast decision problem is strongly NP-hard, and finally that both problems are in NP.

\begin{lemma}\label{th:NP-hard-graph}
The {\cccast} decision problem is strongly NP-hard for trees.  
\end{lemma}

\begin{proof}
We construct a polynomial-time many to one reduction from the following strongly NP-Complete problem \cite{GJ79}.

\probleme{3-Partition}{a multiset $S$ of $3 m$ positive integers $x_i$
  such that for $1\leq i\leq 3m, R/4<x_i<R/2$ with
  $R=\frac{\sum_{i=1}^{3m}x_i}{m}$.}{Can $S$ be partitioned into $m$
  disjoint sets $S_1, S_2,\dots,S_m$ of size 3, such that $\sum_{x\in S_j}x=R$ for
  $1\leq j\leq m$ ?}

We construct an instance $(G,U)$ of the {\cccast} problem from an instance of 3-Partition as follows. The graph $G$ is a star with $4m+2$ leaves and $U$ is the set of leaves of $G$. Hence, there are $4m+2$ agents, each located at a leaf of the star.
We consider a partition of the set of agents into three subsets: $A$, $B$ and $C$. The subset $A=\{a_i\mid 1\leq i\leq m+1\}$ contains $m+1$ agents. The leaves containing these agents are incident to an edge of weight $1$. The subset $B=\{b_i\mid 1\leq i\leq 3m\}$ contains $3m$ agents. For  $1\leq i\leq 3m$, the weight of the edge incident to the leaf containing agent $b_i$ is $2R+1+x_i$. The subset $C=\{c\}$ contains one agent. The leaf containing agent $c$ is incident to an edge of weight $4R+1$. Figure~\ref{fig:reduc} depicts the star obtained in this way. The battery power $P$ allocated to each agent is equal to $2R+1$. The construction can be done in polynomial time. We show that the constructed instance of the {\cccast} problem gives answer yes if and only if the original instance of 3-partition gives answer yes.

\begin{figure}[h]
\centering
\includegraphics[width=10cm]{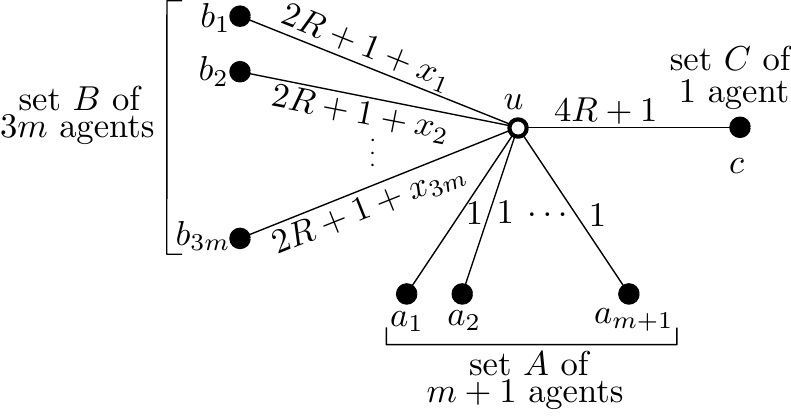}
\caption{Instance of the {\cccast} decision problem constructed from an instance of 3-partition.}
\label{fig:reduc}
\end{figure}

First, assume that there exists a solution $S_1, S_2,\dots,S_m$ for the instance of the 3-partition problem. We show that the agents can solve the corresponding instance of the centralized {\convergecast} problem using the following {\strat}. Agent $c$ moves at distance $2R$ from the center and for each $1\leq i\leq 3m$, agent $b_i$ moves at distance $x_i$ from the center. At this point, all these agents have used all their battery power. Each agent in $A$ moves to the center of the star. For $1\leq i\leq m$ and for each of the three agents $b_j$ such that $x_j\in S_i$, agent $a_i$ moves to meet $b_j$ and goes back to the center of the star. The cost of this movement is $2\sum_{x_j\in S_i}=2R$, which is exactly the remaining battery power of agent $a_i$. Observe that since agents in $A$ have met all agents in $B$, agents in $A$, located at the center of the star, have the information of all agents except agent $c$. Then agent $a_{m+1}$ moves to meet agent $c$. 
Agents $a_{m+1}$ and $c$ have the information of all the agents. Hence, this is a solution of the instance of the {\cccast} problem.

\medskip

Now assume that there exists a solution (strategy) to the convergecast problem.
By Lemma~\ref{cl:simpl}, we can assume that the {\ccast} {\strat} is simple. Consider the star $G$ after the gathering phase of the simple {\strat}. Each agent in $A$ is at the center of the star. For $1\leq i\leq m+1$, the agent $a_i$ has the remaining power of $2R$. For $1\leq i\leq 3m$, the agent $b_i$ is at distance $x_i$ from the center of the star and agent $c$ is at distance $2R$ from the center. Since the agents in $A$ are the only agents with remaining battery power, they must move to collect the information of agents in $B\cup C$. We call this phase the collecting phase. Observe that since agent $c$ is at distance $2R$ from the center, it is impossible for agents in $A$ to transport this information. Indeed, when an agent reaches $c$, it has used all its battery power. Hence, the entire information must be collected at the position of $c$. In order to collect the information, agents in $A$ must go to the position of each agent in $B$ and transport the information of these agents to the center. The total cost to move these information is at least twice the sum of the distances between each agent in $B$ and the center. This is equal to $2\sum_{i=1}{3m}x_i=2Rm$. Then, this information must be moved to the position of $c$. This costs at least $2R$. Hence, the total cost of collecting information after the gathering phase is at least $2R(m+1)$. The amount of power available to the agents for the collecting phase is equal to the amount of power needed to collect the information, since there are $m+1$ agents each having power $2R$. This means that during the collecting phase, for $1\leq j\leq 3m$, agents cannot collectively use a power larger than $2x_i$ to collect the information of $b_i$. 

Suppose by contradiction that during the collecting phase, more than one agent in $A$ enters an edge $f$ to collect the information of agent $b_i$ at distance $x_i$ from the center, for some $i$ such that $1\leq i\leq 3m$.
Let $w$ be the agent that has reached the position of $b_i$. If $w$ comes back to the center, it has used at least power  $2x_i$. Since at least one other agent has used some power to enter edge $f$, these agents have used more than $2x_i$ battery power to collect information of agent $b_i$. If $w$ does not come back to the center, then some other agent has to move the information to the center. If the agent $w$ stops at distance $r$ from the center, then at least one other agent has to go to this position (at distance $r$ from the center) and come back. Thus, the cost is at least $(2x_i-r)+2r>2x_i$. In both cases, the agents have used more power than $2x_i$, which leads to a contradiction. Hence, for each $1\leq i\leq 3m$, there is only one agent that collects the information of agent $b_i$ and enters the corresponding edge. 

We can assume, without loss of generality, that agent $a_{m+1}$ is the agent that transports the information to $c$. Observe that $a_{m+1}$ cannot collect information from other nodes since moving to $c$ uses exactly all its remaining power. Hence, only agents in $A'=A\setminus \{a_{m+1}\}$ can collect the information of agents in $B$. Let $S_1,S_2, \dots, S_m$ be the partition of $S$ defined by $S_i=\{x_j\mid \mbox{the information of $b_j$ is collected by $a_i$}\}$, for each $1\leq i\leq m$. We have $2\sum_{x\in S_i}x\leq 2R$ since each agent from $A'$ has battery power at most $2R$. The power needed to collect information of agents in $B$ is $2mR$ which is exactly equal to the combined power available to agents in  $A'$. This means that each agent in  $A'$ must use all its power to collect information and $2\sum_{x\in S_i}x=2R$. Hence, $S_1,S_2, \dots, S_m$ is a solution to the instance of 3-partition. 
\end{proof}

\begin{lemma}\label{th:NP-hard-graph-b}
The centralized broadcast decision problem is strongly NP-hard for trees.  
\end{lemma}

\begin{proof}
Again, we construct a polynomial-time many to one reduction from 3-Partition. The general structure of the proof is similar as in Lemma~\ref{th:NP-hard-graph} but details differ.

We construct an instance $(G,U)$ of the centralized broadcast problem from an instance of 3-Partition as follows. The graph $G$ is a star with $5m$ leaves and $U$ is the set of leaves of $G$. Hence, there are $5m$ agents, each located at a leaf of the star.
We consider a partition of the set of agents into three subsets: $A$, $B$ and $C$. The subset $A=\{a_i\mid 1\leq i\leq m\}$ contains $m$ agents. The leaves containing these agents are incident to an edge of weight $1$. The subset $B=\{b_i\mid 1\leq i\leq 3m\}$ contains $3m$ agents. For  $1\leq i\leq 3m$, the weight of the edge incident to the leaf containing agent $b_i$ is $4R+1+x_i$. The subset $C=\{c_i\mid 1\leq i\leq m\}$ contains $m$ agents. All leaves containing an agent in $C$
are incident to an edge of weight $6R+1$. Figure~\ref{fig:reducb} depicts the star obtained in this way. The battery power $P$ allocated to each agent is equal to $4R+1$ and agent $a_1$ is the source agent. The construction can be done in polynomial time. We show that the constructed instance of the centralized broadcast problem gives answer yes if and only if the original instance of 3-partition gives answer yes.

\begin{figure}[h]
\centering
\includegraphics[width=14cm]{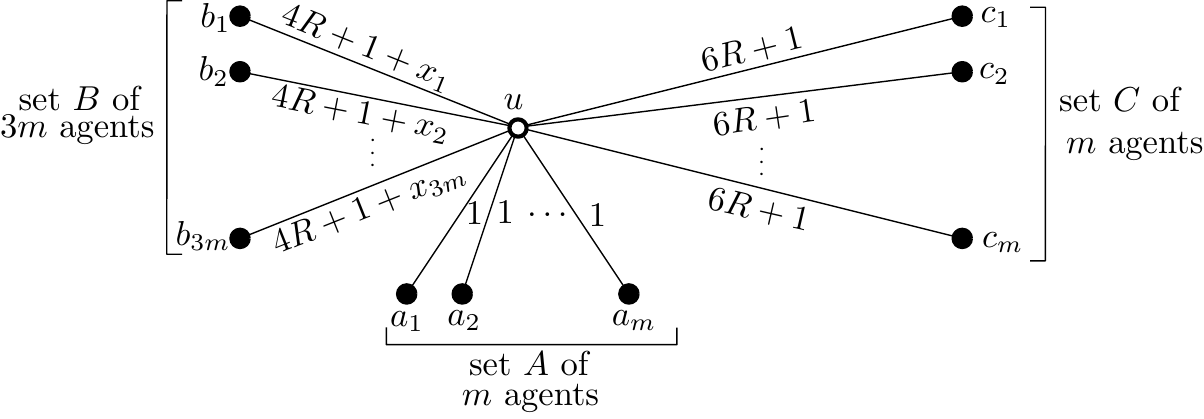}
\caption{Instance of centralized broadcast problem from an instance of 3-partition.}
\label{fig:reducb}
\end{figure}

First, assume that there exists a solution $S_1, S_2,\dots,S_m$ for the instance of the 3-partition problem. We show that the agents can solve the corresponding instance of the centralized broadcast problem using the following {\strat}. For each $i$, agent $c_i$ moves at distance $2R$ from the center and for each $1\leq i\leq 3m$, agent $b_i$ moves at distance $x_i$ from the center. At this point, all these agents have used all their battery power. Each agent in $A$ moves to the center of the star.
Hence, each agent $a_i$ obtains the information of $a_1$. For $1\leq i\leq m$ and each of the three agents $b_j$ such that $x_j\in S_i$, agent $a_i$ moves to meet $b_j$ and goes back to the center of the star. The cost of this movement is $2\sum_{x_j\in S_i}=2R$. Observe that since agents in $A$ have met all agents in $B$, all agents except those in $C$ have the information of $a_1$. Each agent $a_{i}$ moves to meet agent $c_i$. 
Each agent $c_i$ obtains the information of $a_1$. Hence, this is a solution to the instance of the centralized broadcast problem.

\medskip

Now assume that there is a solution (strategy) to the broadcast problem.
By Lemma~\ref{cl:simpl}, we can assume that the centralized broadcast {\strat} is simple. Consider the star $G$ after the gathering phase of the simple {\strat}. Each agent in $A$ is at the center of the star. For $1\leq i\leq m$, the agent $a_i$ has the remaining power of $4R$. For $1\leq i\leq 3m$, the agent $b_i$ is at distance $x_i$ from the center of the star. For $1\leq i\leq m$, agent $c_i$ is at distance $2R$ from the center. Since the agents in $A$ are the only agents with remaining battery power, they must move to give the information to agents in $B\cup C$. Observe that since each agent $c_i$ is at distance $2R$ from the center, an agent in $A$ that moves to meet an agent $c_i$ has not enough power to meet another depleted 
agent afterwards. Hence, each agent $a_i$ must meet exactly one agent $c_j$. Without loss of generality, we can assume that each agent $a_i$ meets $c_i$. Before agents in $A$ meet agents in $C$, they must meet agents in $B$. The total cost to give the information to all agents in $B$ is at least twice the sum of the distances between each agent in $B$ and the center. This is equal to $2\sum_{i=1}{3m}x_i=2Rm$. The total cost to give the information to agents in $C$ is 
$2Rm$. The amount of power available to the agents in $A$ is $4Rm$, which is exactly the power needed for broadcast.
Assume for the sake of contradiction that two or more agents in $A$ enter the same edge incident to the leaf of an agent $b_i$. In this case, one of the agents must meet $b_i$. This costs the agent $2x_i$ and other agents have used some power to enter this edge. This gives a contradiction because the total cost is more than the available power. Thus, we can assume that each agent $b_j$ meets exactly one agent $a_i$. Let $S_1,S_2, \dots, S_m$ be the partition of $S$ defined by $S_i=\{x_j\mid \mbox{$b_j$ met $a_i$}\}$, for each $1\leq i\leq m$. We have $2\sum_{x\in S_i}x\leq 2R$ since the total power that agents  in $A$ can use to meet agents in $B$ is at most $2Rm$. The power needed to give information to agents in $B$ is $2mR$ which is exactly equal to the combined power available to agents in $A$. This means that each agent in $A$ must use all its power to meet agents in $B$ and $2\sum_{x\in S_i}x=2R$. Hence, $S_1,S_2, \dots, S_m$ is a solution to the instance of 3-partition.
\end{proof}

\begin{lemma}\label{lem:NP-graph}
The {\cccast} decision problem and the centralized broadcast decision problem are in NP.  
\end{lemma}

\begin{proof}
We consider the verifier-based definition of NP. Consider the {\strat} $\cA$ of the agents for an instance of the {\cccast} or centralized broadcast problems. We construct the certificate for the instance as follows. We say that a meeting of two or more agents is \emph{useful} if at least one of the agents received a new piece of information during this meeting. Each agent participates in at most $k-1$ useful meetings where $k$ is the number of agents. Hence, there are at most $k(k-1)$ useful meetings. The certificate contains the list of all useful meetings in chronological order. For the $i$-th meeting, the certificate encodes the identities of the meeting agents and the location of the meeting: a node $x_i$  or an edge $(u_i,v_i)$ of the graph $G$. If the meeting has occurred on an edge, the certificate encodes a variable $d_i$. The variable  $d_i$ represents the distance between $u_i$ and the meeting point $p_i$. If a previous meeting of number $j$ has occurred on the same edge, the certificate encodes if $d_i<d_j$, or $d_i=d_j$ or $d_i>d_j$. For each of the meeting agents, the certificate also encodes the node from which it has entered the edge ($u_i$ or $v_i$) just before the meeting and the node from which it exits the edge just after the meeting. We consider the {\strat} $\cA'$ defined as follows. For each useful meeting in chronological order, the meeting agents move to the meeting location following a shortest path from their previous position. If the meeting occurs on an edge, the meeting agents enter and exit the edge using the node encoded in the certificate. $\cA'$ is a {\ccast} {\strat} since each time an agent has collected a new piece of information in $\cA$, it collects the same information during the corresponding meeting in $\cA'$. Moreover, the agents use at most as much power in $\cA'$ as in $\cA$ since they move to the same meeting points using shortest paths. The verifier simulates the {\strat} $\cA'$ defined by the certificate. The verifier first checks that all the agents possess the entire information at the end of the algorithm. This can be done in polynomial time. Then, the verifier computes the distance traveled by each agent. These distances are linear sums of variables $d_i$ with $1\leq i \leq k(k-1)$ and of a constant. Finding an assignment of the variables, such that the distance traveled by each agent is less or equal than $P$, can be done in polynomial time using linear programming. Thus, the certificate can be verified in polynomial time.
\end{proof}

Theorem \ref{th:NP-graph} is a direct consequence of Lemmas \ref{th:NP-hard-graph}, \ref{th:NP-hard-graph-b} and \ref{lem:NP-graph}.

Since both decision problems concerning convergecast and broadcast are NP-hard for the class of trees, the same is true for their optimization counterparts, i.e., computing the smallest amount of power that is sufficient to achieve convergecast or broadcast.
In spite of that, we will show how to obtain, in polynomial time, a 2-approximation of the power needed to achieve {\cccast} on arbitrary graphs and a 4-approximation of the power needed to achieve centralized broadcast on arbitrary graphs. 

Let $D(G,A)=\max_{\emptyset\subsetneq X \subsetneq A}\{ \min_{x\in X, y\in A\setminus X}\{d_G(x,y)\}\}$, where $d_G(x,y)$ is the distance between $x$ and $y$ in $G$. The following proposition shows a relation between $D(G,A)$ and the above optimal power values.

\begin{proposition}\label{lem:twice-p}
Consider a configuration $(G,A)$ for convergecast and a configuration $(G,A)$ with a specified source agent for broadcast. Then $D(G,A)\leq 2P_{OPT}^c$ and $D(G,A)\leq 2P_{OPT}^b$ for any source agent in $(G,A)$.
\end{proposition}

\begin{proof}
We prove the proposition for the case of convergecast. The proof for broadcast is similar.
Suppose, by contradiction, that there is a partition of $A$ into $X$ and $A\setminus X$ such that for each $x\in X$ and $y\in A\setminus X$ the distance between $x$ and $y$ is greater than $2P_{OPT}^c$. It means that no agents in $X$ can meet an agent in $A\setminus X$ using power $P_{OPT}^c$. This contradicts the fact that there is a {\ccast} {\strat} in $G$ using battery power $P_{OPT}^c$. Hence, for every partition of $A$ into $X$ and $A\setminus X$, there exist agents $x\in X$ and $y\in A\setminus X$ that are at distance at most $2P_{OPT}^c$.
\end{proof}

In view of Proposition \ref{lem:twice-p}, the following theorem shows that the convergecast problem has a polynomial-time 2-approximation.

\begin{theorem}\label{cor:FourApr}
Consider a configuration $(G,A)$. There is a polynomial algorithm computing a {\ccast} strategy in which each agent uses power $D(G,A)$.
\end{theorem}

\begin{proof}
We formulate algorithm $\STree$ which produces the desired convergecast strategy.
The parameters of the algorithm are the graph $G$ and the nodes corresponding to the initial positions of agents (stored in $A[1:k]$). 

\medskip

\begin{algorithm}[H]
\SetKw{EmptyStack}{empty\_stack}
\caption{KnownGraph(a weighted graph G, an array A{[1:k]} of nodes) \label{algo:graph}}
\SetKwData{MyTrue}{true} \SetKwData{MyFalse}{false} 
\SetKw{Not}{not} 
\SetKw{MyAnd}{and}
\SetKw{Wait}{Wait}
\SetKw{Move}{Move}
\SetKw{Push}{push}
\SetKw{Pop}{pop}
$strategy=\EmptyStack$\;
$V:=\{A[1]\}$\;
$P:=0$\;
\Repeat{$V= A$}
{
	choose a couple $(u,v)\in V\times(A\setminus V)$ such that $d(u,v)$ is minimal\;
	$V:=V\cup \{v\}$\;
	$Path:=$ shortest path between $u$ and $v$\;
	$\Push(strategy,(v,Path,u))$\;
	$P=max\{P,d(u,v)\}$\;
}
\Repeat{$strategy= \EmptyStack$}
{
	$(v,Path,u)=\Pop(strategy)$\;
	agent starting in $v$ moves to $u$ following path $Path$\;
}
\end{algorithm}

\medskip
Let $(u_i,v_i)$ be the nodes chosen at the $i$-th iteration of the first loop and let $V_i$ be the value of $V$ at the end of the $i$-th iteration. We set $u_0=A[1]$ and $V_0=\{v_0\}$. We show, by induction, that at the start of the $i$-th iteration of the second loop, agents that started in $V_{k-i}$ hold collectively all the information. It is clearly true for $i=1$. Assume by induction that it is true for $i$. The agent that started at $v_{k-i}$ moves to node $u_{k-i}=v_{k-j}$ for some $j>i$, during the $i$-th iteration of the second loop. After this move, the agent that started at $v_{k-j}$ has the information of the agent that started at $v_{k-i}$. Agents in $V_{k-(i+1)}$ collectively hold all the information. Hence, the property is true for $i+1$ and this concludes the argument by induction. At the end of the algorithm, the agent at $A[1]$ has all the information since $V_0=\{A[1]\}$.

Let $A$ be the set of agents. Consider the partition of $A$ into sets $V_{i-1}$ and $A\setminus V_{i-1}$. We have $d(u_i,v_i)\leq D(G,A)$ since $(u_i,v_i)$ is the couple 
$(u,v)\in V_{i-1}\times(A\setminus V_{i-1})$ such that $d(u,v)$ is minimal. Hence, no agent will traverse distance larger than $2P_{OPT}^c$ by Proposition~\ref{lem:twice-p}.

In $O(n^3)$ time, it is possible to precompute all shortest paths between $u$ and $v$ for all $u,v \in A$. Each iteration of the first repeat loop can be computed in $O(n^2)$ time and there are $k-1$ such iterations where $k\leq n$ is the number of agents. Hence, executing the first repeat loop takes time $O(n^3)$. The execution the second repeat loop takes time $O(n^2)$. Hence, the overall complexity of the algorithm is $O(n^3)$.
\end{proof}

\medskip

The above theorem gives the following corollary for the broadcast problem on arbitrary graphs.

\begin{corollary}
The broadcast problem on arbitrary graphs has a polynomial-time 4-approximation.
\end{corollary}

\begin{proof}
Let $(G,A)$ be a configuration with an arbitrary source agent $a$. By Theorem \ref{cor:FourApr}, there is a convergecast strategy $S$ for $(G,A)$ using power at most $D(G,A)$ that can be computed in polynomial time. Let $b$ be the agent that collects all information upon completion of this strategy. Consider the strategy $S'$ which consists of performing the reverse of all moves of $S$ in the reverse order. The strategy $S'$ is a broadcast strategy for source agent $b$. Hence, the strategy $S$ followed by $S'$ is a broadcast strategy for source agent $a$. The required power is at most $2D(G,A)$ which gives a 4-approximation of the broadcast problem in view of Proposition~\ref{lem:twice-p}.
\end{proof}

\section{{\Dccast} and broadcast on trees}\label{s:online}

In this section, we consider the convergecast and the broadcast problem in the distributed setting. 
As explained in the introduction, we consider weighted trees with agents at every leaf. In view of Proposition~\ref{lem:twice-p}, the following theorem implies that there exists a 2-competitive distributed algorithm for the {\ccast} problem on trees.

\begin{theorem}\label{thm:FourComp}
Consider a configuration $(T,A)$ where $T$ is a tree and $A$ contains all the leaves of $T$. There exists a distributed {\ccast} algorithm in which each agent uses power at most $D(T,A)$.
\end{theorem}

\begin{proof}
The idea behind the algorithm is similar to the saturation technique used for message passing systems (see chapter 2.6.1 of \cite{San}). Each agent starting at a leaf moves until it reaches the neighbor of its starting position. When an agent reaches a node, it waits until an agent has arrived from each incident edge except one. When this happens, the agent with the most remaining power moves via the edge from which no agent has arrived. One can show that each agent will not move more than $D(T,A)$ and thus twice $P_{OPT}^c$ by Proposition~\ref{lem:twice-p}. At some point, the saturation occurs, i.e., two agents meet inside an edge or agents meet at a node coming from all incident edges. At this point, the convergecast is achieved.

The pseudocode of the algorithm (executed distributedly by all agents) is the following.

\begin{algorithm}[H]
\caption{UnknownTree \label{algo:2appr}}
\SetKwData{MyTrue}{true}
 \SetKwData{MyFalse}{false} 
\SetKw{Not}{not} 
\SetKw{MyAnd}{and}
\SetKw{Wait}{Wait}
\SetKw{Move}{Move}
 $collecting=\MyFalse$\;
 \While{$collecting=\MyFalse$}
 {\Wait{\emph{}until there is at most one port unused by an agent incoming at the current node}\;
 	\uIf{all ports of the current node were used by incoming agents}
	{$collecting=\MyTrue$}
	  	\uIf{the agent has used less power than any other agent present at the node \MyAnd $collecting=\MyFalse$}
	{\Move{\emph{}through the unused incoming port until you meet another agent or reach a node}\;
	}
	\lElse
	{$collecting=\MyTrue$} 
	\lIf{the agent is inside an edge}{$collecting=\MyTrue$}
 }
\end{algorithm}

 First, we show that if each agent executes Algorithm~\ref{algo:2appr} then, eventually, one agent will hold all the information. Consider an agent $a$ executing the algorithm. Let $T_a(t)$ be the subtree rooted at the last visited node and containing all nodes accessible from the current position of $a$ by shortest paths containing a non-null part of the last edge traversed by agent $a$. Hence, when $a$ enters a new node $u$, $u$ is added to $T_a$. We show by induction on the number of nodes of $T_a$ that $a$ has the initial information of every agent that started at a node of $T_a$. For $|T_a|=1$, this is true since $a$ is the only agent that started in $T_a$. The size of $T_a$ grows only when $a$ enters or exits some node $v$. When $a$ enters a new node $v$, we show that any agent that started at $v$ did not move yet. Assume by contradiction that there is an agent $b$ that started at $v$ and has moved before the arrival of $a$. It means that agents have arrived from all but one edge incident to $v$. In that case, agent $b$ follows the edge from which no agent has arrived. Hence, the only possible edge that agent $b$ can follow is the edge taken by agent $a$ to arrive at $v$. This leads to a contradiction since agents $a$ and $b$ must have met inside the edge and agent $a$ would have stopped before reaching $v$.
 
  When an agent $a$ moves from a node $v$ of degree $\delta$, there were $\delta - 1$ agents $b_1,b_2,\dots,b_{\delta - 1}$ that have arrived at $v$ before. By the induction hypothesis, each agent $b_i$, for $1\leq i\leq \delta-1$, has collected all the information from agents starting inside the subtree $T_{b_i}$. Since agent $a$ moves in the only direction from which no agent has arrived, it has the information of every agent that started in $T_{a}=T_{b_1}\cup T_{b_2}\cup\dots\cup T_{b_{\delta - 1}}$. This concludes the proof by induction.

Observe that for each $1\leq i\leq k$  the tree $T_{a_i}$ grows until either $a_i$ meets agents that have arrived from all incoming ports of its current position, or another agent $a_j$ with more power moves in a yet unexplored direction. In the latter case, $T_{a_i}\subseteq T_{a_j}$ and the tree $T_{a_j}$ will grow under the same conditions. Thus, $\cup_{i=1}^{k}T_{a_i}$ will eventually be equal to $T$. This happens when either two agents $u_1,u_2$ meet inside an edge or $\delta$ agents $u_1,u_2,\dots,u_\delta$ meet at a node of degree $\delta$. These agents have the entire information since  $T_{u_1}\cup T_{u_2}\cup\dots\cup T_{u_\delta}=T$ ($\delta =2$ if the meeting occurs on an edge). 

It remains to show that the agents do not use more battery power than $D(T,A)$. 
Let $p$ be the point where some agent $a$ has finished the execution of the algorithm (when the value of $collecting$ becomes true for this agent) and let $v$ be the last node visited by $a$ before reaching $p$. Consider $T_a$ when $a$ exited $v$. Agent $a$ is the agent starting in $T_a$ for which the distance between its initial position and the node $v$ was the smallest, since it was the agent that has used the least power when it arrived at $v$.  Thus, the distance between the initial position of an agent in $T_a$ and an agent in $G\setminus T_a$ is less or equal than $D(G,A)$. Hence we conclude that $P\leq D(T,A)$. By Property~\ref{lem:twice-p}, we have that $D(T,A)\leq 2P_{OPT}^c$ and hence the algorithm is 2-competitive.
\end{proof}

Again in view of Proposition~\ref{lem:twice-p}, the following corollary implies that there exists 	a 4-competitive distributed algorithm for the broadcast problem on trees.

\begin{corollary}\label{cor:FourComp}
Consider a configuration $(T,A)$ with a specified source agent, where $T$ is a tree and $A$ contains all the leaves of $T$. There exists a distributed broadcast algorithm in which each agent uses power at most $2D(T,A)$.
\end{corollary}

\begin{proof}
Let $(T,A)$ be a configuration with a specified source agent $a$. All agents execute the following algorithm consisting of two phases.
In the first phase, each agent executes algorithm \TwoAprTree from the proof of Theorem \ref{thm:FourComp} to achieve convergecast. Suppose that $B$ is the set of agents that get the total information at the end of the execution of this phase. All agents in $B$ are aware of this fact. Agents in $B$ start the second phase. We call them \emph{active} agents. 
Each active agent backtracks to its initial position, by walking along the path reverse to the one used in phase 1. On its way, it 
activates all agents it meets and conveys all the information to each of them. The process continues until each agent is activated and is back at its initial position. At this time, all information and in particular information of the source agent $a$ is known to all agents. The energy spent is at most $2D(T,A)$.
\end{proof}

The following theorem shows that no distributed algorithm may offer a
better competitive ratio than $2$ for convergecast or for broadcast, even if we only consider line networks. 

\begin{theorem}\label{thm:TwoAprOpt}
Consider any $\delta > 0$, and any value of power $P$. There exists an
integer $n$ and a configuration $Pos[1:n]$ of $n$ agents on the line
such that  :
\begin{itemize}
\item there exists a centralized {\convergecast} strategy using power $P$ and
there is no deterministic distributed strategy allowing the
agents to solve {\convergecast} when the amount of power given to each
agent is $(2-\delta)P$.
\item there exists a centralized {\broadcast} strategy using power $P$ for source agent starting at $Pos[1]$ 
and there is no deterministic distributed strategy for source agent starting at $Pos[1]$  allowing the
agents to solve {\broadcast} when the amount of power given to each
agent is $(2-\delta)P$.
\end{itemize}
\end{theorem}

Before proving Theorem~\ref{thm:TwoAprOpt}, we prove two technical lemmas.  

\begin{lemma}\label{lem-offline-group}
Consider any $\varepsilon > 0$, an amount of power $P$, and a
set $\{a_1, a_2, \ldots,$ $ a_k, a_{k+2}\}$ of $k+2$ agents located at
positions $Pos[1:k+2]$.  If $Pos[k+2]- \rclr(1,P) \leq P-\varepsilon$,
and if $k \geq \log(P/\varepsilon)$, there exists $i \leq k+1$ such
that $\rclr(i,P)\geq Pos[i+1]$.
\end{lemma}

\begin{proof}
Suppose, by contradiction, that the lemma does not hold. It means that
for each $1 \leq i \leq k+1$, $\rclr(i,P) <
Pos[i+1]$. Therefore, in view of the claim from the proof of Lemma~\ref{lem-eqn-reach}, we have
\begin{eqnarray*}
\rclr(k+1,P) &=& 2^{k}\rclr(1,P)+(2^{k}-1)P - \Sigma_{i=1}^{k}
2^{k-i}Pos[i]\\
&=& \rclr(1,P) +(2^{k}-1)P - \Sigma_{i=1}^{k}
2^{k-i}(Pos[i]-\rclr(1,P))\\
&\geq& \rclr(1,P) +(2^{k}-1)P - \Sigma_{i=1}^{k}
2^{k-i}(P-\varepsilon)\\
&\geq& \rclr(1,P) +(2^{k}-1)P - (2^{k}-1)(P-\varepsilon)\\
& \geq& \rclr(1,P) +(2^{k}-1)\varepsilon
\end{eqnarray*}

Consequently, if $k \geq  \log(P/\varepsilon)$, we have 
$\rclr(k+1,P) \geq \rclr(1,P) + P - \varepsilon \geq Pos[k+2]$, 
a contradiction.
\end{proof}

\begin{lemma}\label{lem-online-gap}
Consider an amount of power $P$, a distance $d > 0$, and a set
$\{a_1,a_2, \ldots, a_k\}$ of $k$ agents located at positions
$Pos[1:k]$.  Let $R_1$ be the closest point from $Pos[2]$ that $a_1$
reached. Assume that $Pos[2]-R_1 = d$.

Suppose that all the agents execute the same distributed deterministic algorithm
and do not know their initial position, and assume that some agent $a
\in \{ a_2, a_3\ldots, a_k\}$ meets agent $a_1$ before any couple of agents 
in $\{a_1, a_2, \ldots, a_k\}$ meet. Then, $a = a_2$
and when $a_2$ meets $a_1$, for each $2 \leq i \leq k$, agent $a_i$ is
located on $Pos[i]-d$.

Moreover, if $R_{max}$ is the rightmost point reached by
some agent knowing the initial information of agent $a_1$, then
$R_{max} \leq   Pos[k] + P - 2d$.
\end{lemma}

\begin{proof}
Since all agents are executing the same distributed deterministic algorithm, let
us consider the execution of the algorithm until some agent meets
agent $a_1$. During this period, all the agents perform exactly the
same moves. Since they started simultaneously, no agent meets another agent before agent $a_2$
meets $a_1$ at point $R_1$ or to the left of $R_1$. When agent $a_2$
meets $a_1$, it has moved at least a distance of $d$. Until this
meeting between $a_1$ and $a_2$, every other agent has also moved a
distance of at least $d$, and is located at distance $d$ to the left
of its starting position. Consequently, no agent can go further than
$P-2d$ to the right of $Pos[k]$.
\end{proof}

\begin{proofof}{Theorem \ref{thm:TwoAprOpt}}
Let $\varepsilon = \delta P/4$ and $\sigma = \varepsilon /2= \delta
P/8$. Let $l = \lfloor \log(8/\delta)\rfloor$, $k = l+2$ and $n=2l(l+2)+2$.

Consider a set of $n$ agents positioned on a line as follows (See
Figure~\ref{fig-2-comp}).  There is an agent $a_1$
(resp. $a_{n}$) on the left (resp. right) end of the line on
position $s'_0=0$ (resp. $s_{2l+1}=\ell$). For each $1 \leq i \leq 2l$,
there is a set $A_i$ of $k$ agents that start on distinct positions
within a segment $[s_i,s'_i]$ of length $\sigma$ such that for each $1
\leq i \leq 2l+1$, the distance between $s_i$ and $s'_{i-1}$ is
$2(P-\varepsilon)$. In other words, for each $i$, $s_i=
(2P-3\sigma)i-\sigma$ and $s_i'=(2P-3\sigma)i$.

\begin{figure}
\centering
\includegraphics[width=16cm]{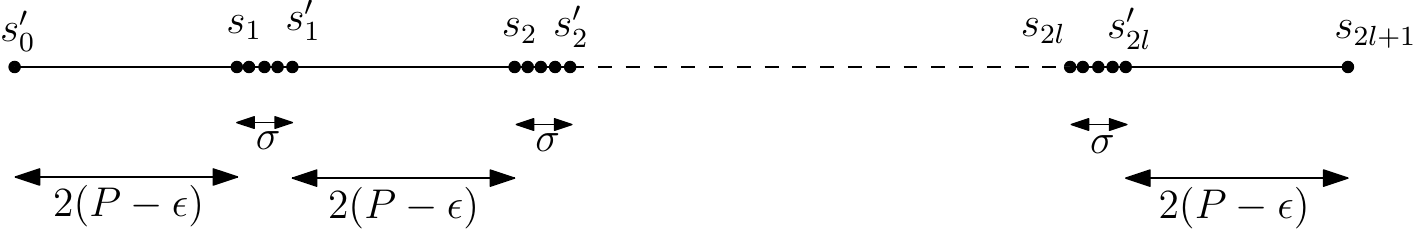}%
\caption{The configuration in the proof of Theorem~\ref{thm:TwoAprOpt}.}
\label{fig-2-comp}
\end{figure}

First, let us consider the execution of the optimal convergecast centralized algorithm for this
configuration. We claim that if the amount of power given to each
agent is $P$, then {\convergecast} is achievable. We show by induction on
$i$ that for every $i$, $\rclr(ik+1,P) \geq s'_i + P-\epsilon = s_{i+1}
- P +\epsilon$. For $i = 0$, $\rclr(1,P) = Pos[1]+P > P-\epsilon =
s_1-P+\epsilon$. Suppose that $\rclr((i-1)k+1,P) \geq s_{i} - P
+\epsilon$. Consider the agents in $A_i$, i.e., the agents $a_{ik+1-j}$,
$j \in [0,k-1]$. Since $s'_i - \rclr((i-1)k+1,P) \leq
P-\epsilon+\sigma = P-\sigma < P$, and since $l+1 \geq \log(P/\sigma)$,
we know by Lemma~\ref{lem-offline-group} that
$\rclr(k(i-1)+l+2,P)\geq Pos[k(i-1)+l+3]$. Since $k > l+1$, it
follows that $\rclr(ik,P) = Pos[ik]+P \geq s_i + P =
s_i'+P-\sigma \geq s_i'+P-\epsilon$. Consequently, this concludes the proof by induction. Since $\rclr(2lk+1,P) \geq s_{2l+1} - P
+\epsilon \geq \rcrl(2lk+2,P)$, $P$ is sufficient to solve
{\convergecast}. Notice that the same strategy guarantees broadcast for source agent $a_1$ for configuration 
$Pos[1:n]$ and power $P$.

\medskip

Consider any distributed deterministic algorithm where the amount of power
given to each agent is $(2-\delta)P$, yielding a strategy $\cS$ of the agents. 
 A \emph{step} in $\cA$ is a moment when two
agents meet.  Let $t_{l,i}$ (resp. $t_{r,i}$) be the first step where an
agent from $A_i$ meets an agent from $A_{i'}$ with $i' < i$ (resp. $i'
> i$).  Let $R_i$ (resp. $L_i$) be the rightmost point
(resp. the leftmost point) reached by any agent from $A_i$ after some agent in
$A_i$ has met an agent from $A_{i'}$ with $i' < i$ (resp. $i' >
i$). For any $1\leq i < j \leq 2l+1$, let $A_{i,j}= A_i \cup A_{i+1}
\ldots  \cup A_{j}$.

We show by induction on time $t$ that for each $i \in [1,l]$ such that
$t_{l,i} \leq t$ and for each $j \in [l+1,2l]$ such that $t_{r,j} \leq
t$, the following properties hold:
\begin{enumerate}[$(i)$]
\item $t_{l,i} < t_{l,i'}$ for each $i' \in [i+1,l]$ and $t_{r,j} < t_{r,j'}$
  for each $j' \in [l+1,j-1]$, 
\item for each $i' \in [i+1,l]$, if $t_{l,i'} > t$ then $t_{r,i'} > t$, and
  for each $j' \in [l+1,j-1]$, if $t_{r,j'} > t$ then $t_{l,j'} > t$ 
\item $R_i \leq s_{i+1}-(2^{i+2}-2)\varepsilon$ and  $L_{j} \geq
s'_{j-1} + (2^{2l-1-j}-2)\varepsilon$,
\item no agent in $A_{i'}$, $i' \geq l$ meets
  any agent from $A_{1,l-1}$ and no agent in $A_{j'}$, $j'\leq l+1$
  meets any agent from $A_{l+2,2l+1}$.
\end{enumerate}

First, consider $t=0$.
Clearly, $R_1 \leq s_0' + 2P-\delta P = 2P - 4\varepsilon = s_1 -
2\varepsilon$ and $L_{2l+1} \geq s_{2l+1} - 2P +\delta P =
s_{2l} +2 \varepsilon$. Since all agents in $A_{1,2l}$ execute the
same algorithm, they all perform the same moves until either the
leftmost agent of $A_1$ meets $a_0$ (at step $t_{l,1}$), or the rightmost
agent of $A_{2l}$ meets $a_{2l+1}$ (at step $t_{r,2l}$). In the
first case, it shows that $t_{l,1} < t_{l,i}$ and $t_{l,1} < t_{r,i}$ for any $i\geq 2$.
By Lemma~\ref{lem-online-gap}, $R_1 \leq s_1' +
(2-\delta)P-2(s_1-R_0) \leq s_2-2P+2\varepsilon + 2P -4\varepsilon
-2(2\varepsilon) = s_2- 6\varepsilon$.  By symmetry, in the second case,
$t_{r,2l} < t_{r,i}$ and $t_{r,2l} < t_{l,i}$ for any $i\leq 2l-1$ and $L_{2l} \geq
s'_{2l-1}+6\varepsilon$. In both cases, properties for $(i)-(iii)$ hold for $t=0$.
Notice that for any $i,j\in [1,2l]$, no agent in $A_i$ has met an agent of $A_j$.
Hence, property $(iv)$ hold for $t=0$.

Suppose that the induction hypothesis holds for all $t' < t$ and let $i =
\max \{i' \mid t_{l,i'} < t\}+1$ and $j = \min \{j' +1 \mid t_{r,j'} <
t\} - 1$. Note that by $(iv)$, we have $i \leq l-1$ and $j \geq l+2$.
By $(i)$ and $(ii)$, before step $t$, no agent in $A_{i'}$, $i \leq i'
\leq j$ has met any other agent from a set $A_{i''}$, $i' \neq
i''$. Thus, since all agents in $A_{ij}$ execute the same deterministic distributed algorithm starting simutaneously,
they have performed exactly the same moves and they have not met any
other agent before step $t$.  Suppose that an agent from $A_{ij}$
meets another agent at step $t$. Then, either the leftmost agent $a_i$
from $A_i$ meets an agent $a_{i'}$ from $A_{i'}$ with $i' < i$, or the
rightmost agent from $A_j$ meets an agent from $A_{j'}$ with $j' > j$.

By symmetry, it is enough to consider only one case. In the following,
we assume that $a_i \in A_i$ meets an agent $a_{i'} \in A_{i'}$ with
$i' < i$ at step $t$.  In this case, $t = t_{l,i}$ and thus $t_{l,i} <
t_{l,i'}$ and $t_{l,i} < t_{r,i'}$ for each $i < i'\leq j$; consequently, properties $(i)$ and
$(ii)$ hold for $t$.  Moreover, by induction hypothesis, the meeting between
$a_i$ and $a_{i'}$ occurs at a point $p \leq R_{i'} \leq R_{i-1} \leq
s_{i}-(2^{i+1}-2)\varepsilon$. First suppose that $i \leq l-1$.  By
Lemma~\ref{lem-online-gap}, we have $R_i \leq s'_{i} + 2P -\delta P
-2(2^{i+1}-2)\varepsilon = s_{i+1} - 2P + 2\varepsilon +2P -
4\varepsilon - 2^{i+2}\varepsilon + 4\varepsilon = s_{i+1} -
(2^{i+2}-2)\varepsilon$, and thus property $(iii)$ and $(iv)$ holds for $t$. Then suppose that $i =
l \geq \log(8/\delta)-1$. We have $R_i' \leq s_{l}-(8/\delta-2)\delta
P/4 = s_{l}-2P+\delta P/2 < s_{l}-2P+\delta P$. But this is impossible
since the initial position of the leftmost agent $a$ of $A_{l}$ is
$Pos[lk+1]\geq s_l$ and the power available to $a$ is $2P-\delta P$.
This concludes the proof by induction. In particular, no agent from $A_{1,l}$ ever meets any agent from
$A_{l+1,2l+1}$ and consequently, $\cA$ is neither a distributed
{\convergecast} strategy nor a distributed
{\broadcast} strategy for any source agent.
\end{proofof}

Theorems~\ref{thm:FourComp} and \ref{thm:TwoAprOpt} show that for the distributed convergecast problem on the class of trees, the competitive ratio 2 is optimal. 

\section{Conclusion and open problems}\label{s:open}


In the centralized setting, we showed that the breaking point in complexity between polynomial and NP-hard, both for the convergecast and for the broadcast problem, is already present inside the class of trees. Namely, agents' optimal power and the strategy using it can be found in polynomial time for the class of lines but it is NP-hard for the class of arbitrary trees. Nevertheless, we found polynomial approximation algorithms for both these problems. It remains open if better approximation constants can be found.

The problem of a single {\em information transfer} by mobile agents between two stationary points of the network, which we called \emph{carry} in the case of lines, is also interesting. In particular, it is an open question whether the problem of finding optimal power for this task is NP-hard for arbitrary tree networks or if a polynomial-time algorithm is possible in this case. Our reduction from 3-partition is no longer valid for this problem. 

In the distributed setting, we showed that 2 is the best competitive ratio for the problem of convergecast. However, our distributed algorithm for the broadcast problem is only 4-competitive. It remains open to find the best competitive ratio for the broadcast problem.

Additional natural questions related to our research include
other variations of the agent model, e.g., agents with unequal power, agents with non-zero visibility, labeled agents in the distributed setting, as well as fault-tolerant issues, such as unreliable agents or networks with possibly faulty components.



\begin{thebibliography}{10}

\bibitem{Albers}
S.~Albers.
\newblock Energy-efficient algorithms.
\newblock {\em Communications of the ACM}, 53(5):86--96, 2010.

\bibitem{AH}
S.~Albers and M.~R. Henzinger.
\newblock Exploring unknown environments.
\newblock In {\em Proceedings of the 29th Annual ACM Symposium on
  Theory of Computing (STOC)}, pages 416--425, 1997.

\bibitem{AG}
S.~Alpern and S.~Gal.
\newblock {\em The theory of search games and rendezvous}, volume~55,
\newblock Springer, 2003.

\bibitem{Ambuhl}
C.~Ambühl.
\newblock An optimal bound for the mst algorithm to compute energy efficient
  broadcast trees in wireless networks.
\newblock In {\em Proceedings of the International Colloquium on Automata, Languages, and
  Programming (ICALP)}, volume 3580 of {\em Lecture Notes in Computer Science},
  pages 1139--1150, 2005.

\bibitem{AOSY}
H.~Ando, Y.~Oasa, I.~Suzuki, and M.~Yamashita.
\newblock Distributed memoryless point convergence algorithm for mobile robots
  with limited visibility.
\newblock {\em IEEE Transactions on Robotics and Automation}, 15(5):818--828,
  1999.

\bibitem{AA06}
D.~Angluin, J.~Aspnes, Z.~Diamadi, M.~Fischer, and R.~Peralta.
\newblock Computation in networks of passively mobile finite-state sensors.
\newblock {\em Distributed Computing}, 18(4):235--253, 2006.

\bibitem{AGS}
V.~Annamalai, S.~Gupta, and L.~Schwiebert.
\newblock On tree-based convergecasting in wireless sensor networks.
\newblock In {\em IEEE Wireless Communications and Networking}, volume~3, pages 1942--1947, 2003.

\bibitem{AIS}
J.~Augustine, S.~Irani, and C.~Swamy.
\newblock Optimal power-down strategies.
\newblock In {\em Proceedings of the 45th
  Annual IEEE Symposium on Foundations of Computer Science (FOCS)}, pages 530--539, 2004.

\bibitem{AB}
I.~Averbakh and O.~Berman.
\newblock A heuristic with worst-case analysis for minimax routing of two
  travelling salesmen on a tree.
\newblock {\em Discrete Applied Mathematics}, 68(1–2):17 -- 32, 1996.

\bibitem{ABRS}
B.~Awerbuch, M.~Betke, R.~L. Rivest, and M.~Singh.
\newblock Piecemeal graph exploration by a mobile robot.
\newblock {\em Information and Computation}, 152(2):155 -- 172, 1999.

\bibitem{AGP}
B.~Awerbuch, O.~Goldreich, R.~Vainish, and D.~Peleg.
\newblock A trade-off between information and communication in broadcast
  protocols.
\newblock {\em Journal of the ACM}, 37(2):238--256, 1990.

\bibitem{Azar}
Y.~Azar.
\newblock On-line load balancing.
\newblock {\em A. Fiat and G.Woeginger, Online Algorithms: The State of the
  Art, Lecture notes in computer science}, 1442:178--195, 1998.

\bibitem{BCR}
R.~Baezayates, J.~Culberson, and G.~Rawlins.
\newblock Searching in the plane.
\newblock {\em Information and Computation}, 106(2):234 -- 252, 1993.

\bibitem{BGI}
R.~Bar-Yehuda, O.~Goldreich, and A.~Itai.
\newblock On the time-complexity of broadcast in multi-hop radio networks: An
  exponential gap between determinism and randomization.
\newblock {\em Journal of Computer and System Sciences}, 45(1):104 -- 126,
  1992.

\bibitem{BS}
M.~Bender and D.~Slonim.
\newblock The power of team exploration: two robots can learn unlabeled
  directed graphs.
\newblock In {\em Proceedings of the 35th
  Annual Symposium on Foundations of Computer Science (FOCS)}, pages 75--85, 1994.

\bibitem{BFRSV}
M.~A. Bender, A.~Fernández, D.~Ron, A.~Sahai, and S.~Vadhan.
\newblock The power of a pebble: Exploring and mapping directed graphs.
\newblock {\em Information and Computation}, 176(1):1 -- 21, 2002.

\bibitem{BeRS}
M.~Betke, R.~Rivest, and M.~Singh.
\newblock Piecemeal learning of an unknown environment.
\newblock {\em Machine Learning}, 18(2-3):231--254, 1995.

\bibitem{BlRS}
A.~Blum, P.~Raghavan, and B.~Schieber.
\newblock Navigating in unfamiliar geometric terrain.
\newblock {\em SIAM Journal on Computing}, 26(1):110--137, 1997.

\bibitem{Bunde}
D.~Bunde.
\newblock Power-aware scheduling for makespan and flow.
\newblock {\em Journal of Scheduling}, 12(5):489--500, 2009.

\bibitem{CJABL}
F.~Chen, M.~Johnson, Y.~Alayev, A.~Bar-Noy, and T.~La~Porta.
\newblock Who, when, where: Timeslot assignment to mobile clients.
\newblock {\em IEEE Transactions on Mobile Computing}, 11(1):73--85, 2012.

\bibitem{CFPS}
M.~Cieliebak, P.~Flocchini, G.~Prencipe, and N.~Santoro.
\newblock Solving the robots gathering problem.
\newblock In {\em Proceedings of the International Colloquium of Automata, Languages and Programming (ICALP)},
  volume 2719 of {\em Lecture Notes in Computer Science}, pages 1181--1196.
  Springer Berlin Heidelberg, 2003.

\bibitem{CP}
R.~Cohen and D.~Peleg.
\newblock Convergence properties of the gravitational algorithm in asynchronous
  robot systems.
\newblock {\em SIAM Journal on Computing}, 34(6):1516--1528, 2005.

\bibitem{C15}
A.~Cord-Landwehr, B.~Degener, M.~Fischer, M.~Hüllmann, B.~Kempkes, A.~Klaas,
  P.~Kling, S.~Kurras, M.~Märtens, F.~Meyer auf~der Heide, C.~Raupach,
  K.~Swierkot, D.~Warner, C.~Weddemann, and D.~Wonisch.
\newblock A new approach for analyzing convergence algorithms for mobile
  robots.
\newblock In L.~Aceto, M.~Henzinger, and J.~Sgall, editors, In {\em Proceedings of the International
  Colloquium of Automata, Languages and Programming (ICALP)}, volume 6756 of
  {\em Lecture Notes in Computer Science}, pages 650--661, 2011.

\bibitem{DFSY}
S.~Das, P.~Flocchini, N.~Santoro, and M.~Yamashita.
\newblock On the computational power of oblivious robots: Forming a series of
  geometric patterns.
\newblock In {\em Proceedings of the 29th ACM SIGACT-SIGOPS Symposium on
  Principles of Distributed Computing (PODC)}, pages 267--276, 2010.

\bibitem{DP}
X.~Deng and C.~H. Papadimitriou.
\newblock Exploring an unknown graph.
\newblock {\em Journal of Graph Theory}, 32(3):265--297, 1999.


\bibitem{DKS}
M.~Dynia, M.~Korzeniowski, and C.~Schindelhauer.
\newblock Power-aware collective tree exploration.
\newblock In {\em Architecture of Computing Systems (ARCS)}, volume 3894 of {\em Lecture Notes in
  Computer Science}, pages 341--351, 2006.

\bibitem{FPSW}
P.~Flocchini, G.~Prencipe, N.~Santoro, and P.~Widmayer.
\newblock Gathering of asynchronous robots with limited visibility.
\newblock {\em Theoretical Computer Science}, 337(1–3):147 -- 168, 2005.

\bibitem{FGKP}
P.~Fraigniaud, L.~Gasieniec, D.~R. Kowalski, and A.~Pelc.
\newblock Collective tree exploration.
\newblock {\em Networks}, 48(3):166--177, 2006.

\bibitem{FHK}
G.~Frederickson, M.~Hecht, and C.~Kim.
\newblock Approximation algorithms for some routing problems.
\newblock {\em SIAM Journal on Computing}, 7(2):178--193, 1978.

\bibitem{GJ79}
M.~R. Garey and D.~S. Johnson.
\newblock {\em Computers and Intractability: A Guide to the Theory of
  NP-Completeness}.
\newblock W. H. Freeman \& Co., New York, NY, USA, 1979.

\bibitem{ISG}
S.~Irani, S.~Shukla, and R.~Gupta.
\newblock Algorithms for power savings.
\newblock {\em ACM Transactions on Algorithms}, 3(4), 2007.

\bibitem{KK}
A.~Kesselman and D.~R. Kowalski.
\newblock Fast distributed algorithm for convergecast in ad hoc geometric radio
  networks.
\newblock {\em Journal of Parallel and Distributed Computing}, 66(4):578 --
  585, 2006.
\newblock Algorithms for Wireless and Ad-Hoc Networks.

\bibitem{KEW}
B.~Krishnamachari, D.~Estrin, and S.~Wicker.
\newblock The impact of data aggregation in wireless sensor networks.
\newblock In {\em Proceedings of the 22nd International Conference on Distributed Computing Systems Workshops}, pages 575--578, 2002.

\bibitem{MMS}
N.~Megow, K.~Mehlhorn, and P.~Schweitzer.
\newblock Online graph exploration: New results on old and new algorithms.
\newblock {\em Theoretical Computer Science}, 463:62--72, 2012.

\bibitem{RV}
R.~Rajagopalan and P.~Varshney.
\newblock Data-aggregation techniques in sensor networks: a survey.
\newblock {\em IEEE Communications Surveys \& Tutorials}, 8(4):48--63, 2006.

\bibitem{San}
N.~Santoro.
\newblock {\em Design and analysis of distributed algorithms}, volume~56.
\newblock John Wiley \& Sons, 2006.

\bibitem{SL}
I.~Stojmenovic and X.~Lin.
\newblock Power-aware localized routing in wireless networks.
\newblock {\em IEEE Transactions on Parallel and Distributed Systems},
  12(11):1122--1133, 2001.

\bibitem{SY}
I.~Suzuki and M.~Yamashita.
\newblock Distributed anonymous mobile robots: Formation of geometric patterns.
\newblock {\em SIAM Journal on Computing}, 28(4):1347--1363, 1999.

\bibitem{YS}
M.~Yamashita and I.~Suzuki.
\newblock Characterizing geometric patterns formable by oblivious anonymous
  mobile robots.
\newblock {\em Theoretical Computer Science}, 411(26–28):2433 -- 2453, 2010.

\bibitem{YDS}
F.~Yao, A.~Demers, and S.~Shenker.
\newblock A scheduling model for reduced cpu energy.
\newblock In {\em Proceedings of the 36th Annual Symposium on Foundations of Computer Science}, pages 374--382, 1995.

\end{thebibliography}
\end{document}